\documentclass[preprint,3p,authoryear]{elsarticle}

\usepackage{amsmath,amsfonts,amssymb,url}
\usepackage{bm}
\usepackage[figuresright]{rotating}
\usepackage{multirow}

\newcommand{\cov}{\text{Cov}}
\newcommand{\corr}{\text{Corr}}
\newcommand{\bE}{\text{E}}
\newcommand{\ii}{\mathbf{i}}
\newcommand{\jj}{\mathbf{j}}
\newcommand{\bI}{\mathbf{I}}
\newcommand{\bX}{\mathbf{X}}
\newcommand{\cJ}{\mathcal{J}}
\newcommand{\cN}{\mathcal{N}}

\newcommand{\bP}{\Pr}
\newcommand{\var}{\text{Var}}

\newcommand{\balpha}{\bm{\alpha}}
\renewcommand{\mathbf}[1]{\bm{#1}}

\newcommand{\bzero}{\mathbf{0}}
\newcommand{\bone}{\mathbf{1}}
\newcommand{\ignore}[1]{}

\def\bSig\mathbf{\Sigma}

\newtheorem{assumption}{Assumption}
\newtheorem{condition}{Condition}
\newtheorem{corollary}{Corollary}
\newproof{proof}{Proof}
\newtheorem{proposition}{Proposition}
\newtheorem{remark}{Remark}
\newtheorem{theorem}{Theorem}

\newenvironment{psmallmatrix}{\left(\begin{smallmatrix}}{\end{smallmatrix}\right)}

\usepackage{tikz}

\newif\ifstartedinmathmode
\newcommand\encircled[1]{%
  \relax\ifmmode\startedinmathmodetrue\else\startedinmathmodefalse\fi%
  \tikz[baseline,anchor=base]{%
  \node[draw,circle,outer sep=0pt,inner sep=.2ex]
    {\ifstartedinmathmode$#1$\else#1\fi};}%
}

\begin{document}

\begin{frontmatter}


 \title{A new non-parametric test for multivariate paired data from pair matching or paired designs}
%\tnotetext[label1]{hello}

\author[label1]{Jingru Zhang}
\affiliation[label1]{organization={School of Data Science, Fudan University},
            addressline={Shanghai 200433},
            country={China}}
\author[label2]{Hao Chen\corref{cor1}}
\ead{hxchen@ucdavis.edu}
%\fntext[label2]{}
\cortext[cor1]{Corresponding author}
\affiliation[label2]{organization={Department of Statistics, University of California, Davis},
            addressline={Davis, CA 95616},
            country={USA}}
\author[label3]{Xiao-Hua Zhou}
\affiliation[label3]{organization={Beijing International Center for Mathematical Research and School of Public Health, Peking University},
            addressline={Beijing 100871},
            country={China}}

%\title{}


%\author{}

\begin{abstract}
In observational studies, achieving covariate balance in pair matching between treatment and control groups or exposed and unexposed groups is essential. 
This balance enables testing treatment effects or examining {associations between exposures and} multivariate response variables in pair-matched data. 
Paired design studies involve taking multiple measurements for the same subjects under different conditions. All these call for an effective test for multivariate paired data.
However, current methods for assessing covariate balance in matched observational studies often ignore the paired structure, leading to reduced performance in some cases. The multivariate paired Hotelling's $T^2$ test can be used for paired data, but its power decreases rapidly as dimensions increase.
To address these issues, we propose a new non-parametric test for paired data, significantly improving power across various scenarios. We also derive the test's asymptotic distribution, making it user-friendly for practical applications. Our proposed test's effectiveness is demonstrated through an analysis of real data on Alzheimer's disease research.
\end{abstract}




\end{frontmatter}


\section{{Introduction}}\label{sec:intro}
Random assignment is ideal for analyzing treatment effects or {associations with} exposures, as it ensures that covariates are distributed similarly between groups.
However, randomized studies are often impractical, particularly when studying inherent factors like sex. 
In such cases, observational studies rely on multivariate matching to balance covariates between groups. 
While exact matching is rarely feasible when dealing with many covariates, approximate balance in paired data can be achieved using matching strategies. 
It is then crucial to evaluate whether covariate balance has been adequately achieved.

One common approach to assess covariate balance is through the propensity score \citep{rosenbaum1985constructing}, which is the probability of a subject being assigned to the treatment given covariates.
By regressing treatment assignment on covariates, multivariate covariates are summarized into a propensity score. 
However, this method requires modeling assumptions \citep{austin2008assessing,austin2019assessing,cannas2019comparison,harder2010propensity,imai2014covariate}. 
To reduce dependency on these assumptions, less model-based methods have been proposed. 
For instance, {\citet{gagnon2019classification} developed the classification permutation test, which combines classification methods with Fisherian permutation inference.}
\citet{hansen2008covariate} introduced the method of combined differences.
Graph-based methods have also gained attention for covariate balance testing. 
\citet{rosenbaum2005} proposed the crossmatch test, which evaluates the between-group edge count on a nonbipartite graph, and this method has been applied to testing covariate balance in many studies \citep{de2016evaluation,heller2010using}.  
More recently, \citet{chen2022new} proposed new tests, CrossNN and CrossMST, based on the nearest neighbor graph and the minimum spanning tree, respectively, where two within-group edge counts are used to construct statistics.  
However, these methods do not take pair matching into consideration, potentially reducing their effectiveness. 
For example, in exploring the { association between sex and} Alzheimer's disease \citep{carter2012sex,mazure2016sex}, researchers often match male and female participants in pairs based on a few covariates. It is then important to assess whether these covariates are well balanced after matching.
As we will see from the results in our real application (Section \ref{sec:realapp2}), existing methods may indicate balanced covariates, while the paired $t$-test with Bonferroni correction rejects the null hypothesis of balance.

Paired data are also common in paired design studies, where the same subjects are measured twice under different conditions. 
In such cases, testing whether multivariate outcomes differ significantly between conditions is of scientific interest.  
{The multivariate paired Hotelling's $T^2$ test is widely used in low-dimensional settings \citep{rencher2012methods}. However, its effectiveness diminishes in moderate-dimensional cases, and it faces significant challenges in high-dimensional settings unless strong assumptions are imposed to enable the estimation of the covariance matrix.  On the other hand, two-sample hypothesis testing approaches are generally unsuitable for paired data from paired designs, as such data typically exhibit dependencies between paired observations, whereas two-sample tests  assume independence between the two samples.}
%In the nonparametric two-sample regime, existing graph-based methods \citep{friedman1979,chen2017new,chen2018weighted,chen2022new} assume independence between samples, making them unsuitable for non-independent paired data from paired designs. 

{
To be more specific, let $(X_i, Y_i),~i=1, \dots, n,$ represent the paired data, where $X_i$ and $Y_i$ are multivariate observations.  Assume that $(X_i, Y_i) \stackrel{iid}{\sim} F_1$. 
%We consider the null hypothesis for paired data to be that within-pair observations are exchangeable.
Under the null hypothesis of no difference between the conditions in the paired design, the observations before and after are exchangeable.
That is, $(Y_i, X_i) \stackrel{iid}{\sim} F_2$, and we are interested in testing 
\begin{equation}\label{eq:prob}
H_0: F_1 = F_2,
\end{equation}
against the alternative
\begin{equation*}
H_a: F_1 \neq F_2.
\end{equation*}
On the other hand, in the typical two-sample testing setting, assume that
$X_i \stackrel{iid}{\sim} F_X$, and $Y_i \stackrel{iid}{\sim} F_Y$. 
The goal is to test \begin{equation}\label{eq:testold}
H_0: F_X = F_Y,
\end{equation}
against the alternative
\begin{equation*}
H_a: F_X \neq F_Y.
\end{equation*}
%It is important to note that when $X_i$ and $Y_i$ are independent, the hypothesis in this two-sample testing framework is equivalent to that in the paired testing framework. 
%However, paired data often exhibit dependencies within each pair, making the test problem \eqref{eq:prob} more appropriate. 
It is important to note that the null hypothesis \eqref{eq:prob} is equivalent to requiring that the joint distribution of \((X_i, Y_i)\) be symmetric under coordinate swapping. Consequently, when \(X_i\) and \(Y_i\) are independent, the two hypotheses \eqref{eq:prob} and \eqref{eq:testold} coincide.
However, when data within each pair are dependent---as is often the case with paired data---\eqref{eq:testold} may still hold while \eqref{eq:prob} fails due to the asymmetry in the joint distribution. 
{From a substantive perspective, hypothesis (1) is often the more relevant target in paired or repeated-measures studies because it assesses whether the observations obtained under the two conditions are exchangeable within individuals. In contrast, hypothesis (2) only compares the marginal distributions and ignores the dependence structure within pairs. As a result, systematic within-pair changes may be obscured when positive and negative individual-level effects offset each other at the population level. For example, in a causal inference setting, a treatment may substantially benefit some individuals while adversely affecting others, leading to identical marginal distributions before and after treatment. In such a scenario, hypothesis (2) would not reject despite the presence of meaningful individual-level treatment effects, whereas hypothesis (1) could still detect the lack of exchangeability within pairs. Thus, testing (1) can reveal departures from within-pair exchangeability that would be invisible to analyses based solely on marginal distributions.}
A simple example illustrating this distinction is provided in \ref{toyexample}.

%For instance, while independence is theoretically assumed in matched studies or inferences following matching, empirical dependencies frequently arise due to the inherent paired structure. 
%Consequently, the proposed test for \eqref{eq:prob} often outperforms or performs on par with unpaired tests for the usual two-sample null \eqref{eq:testold} (Section \ref{sec:simu-pairmatch}). Additionally, in paired design studies where the same subjects are measured under different conditions, the non-independence between observations further necessitates the use of \eqref{eq:prob}.
}

To address the testing problem \eqref{eq:prob} for multivariate data, 
we adopt a non-parametric framework -- the graph-based framework. 
Graph-based tests have gained popularity for their flexibility and mild assumptions on data distributions. They can accommodate heavy-tailed or skewed distributions, high-dimensional settings, and a variety of distributional differences, whether global or localized to sparse subsets of variables. 
These tests also tend to outperform likelihood-based methods in high-dimensional settings \citep{friedman1979,chen2017new,chen2018weighted,chen2022new,zhang2022graph}.
However, existing graph-based methods do not account for the paired structure and instead rely on the two-sample permutation framework.
Specifically, with $m$ observations in sample 1 and $n$ observations in sample 2, these methods assume exchangeability under the null hypothesis \eqref{eq:testold} and assign equal probability to each of the $\begin{psmallmatrix}
m+n \\
m
\end{psmallmatrix}$ permutations of assigning $m$ out of $m+n$ pooled observations to sample 1 and the rest to sample 2.
While suitable for unpaired data, this framework fails for paired data, where $X_i$ and $Y_i$ may be dependent. 
Consequently, existing statistics and techniques for deriving analytical results and asymptotic distributions become inapplicable.
To overcome these limitations, we develop a novel permutation framework tailored to paired data, and summarize our contributions as follows:
\begin{enumerate}
\item A new test for assessing covariate balance in pair matching: We propose a non-parametric test that accounts for the paired structure, demonstrating higher power than existing methods for assessing multivariate covariate balance across various settings.

\item Broad applicability to multivariate paired data: The proposed test works effectively for paired design studies, where subjects are measured twice under different conditions, and for matched studies, where multiple response variables assess treatment or exposure effects in well-paired data (e.g., neuropsychological measures in matched male and female participants).

\item A novel paired-comparison permutation framework: We introduce a permutation framework specific to paired data, which differs significantly from the two-sample framework used in prior works \citep[e.g.,][]{friedman1979,chen2017new,chen2018weighted,chen2022new}. This framework introduces complexities in deriving analytical expressions for expectations and variances, as well as in proving asymptotic properties, requiring novel theoretical tools.
\end{enumerate}

The remainder of the paper is organized as follows. Section \ref{sec:method} presents the proposed non-parametric framework for paired data and derives the asymptotic distribution of the test statistic under the paired-comparison permutation framework. Section \ref{sec:simu} examines the performance of the proposed test through extensive simulation studies. Section \ref{sec:realdata} demonstrates the test's application in Alzheimer's disease research. Finally, we conclude the paper with a discussion in Section \ref{sec:discuss} and a summary of findings in Section \ref{sec:conclude}.

\section{A new non-parametric test based on a similarity graph}\label{sec:method}
To assess multi-covariate balance from pair matching or multivariate paired data, the one-to-one correspondence between paired observations makes it inappropriate to use the common permutation null distribution applied in the usual two-sample testing setting. 
Instead, it is more suitable to treat the paired observations as exchangeable under the null hypothesis.
This approach is particularly necessary for non-independent paired data from paired designs, where observations from the two samples are no longer exchangeable. 
Let $(X_i, Y_i)$, $i \in \{1, \ldots, n\}$, denote the paired data from pair matching or paired designs. 
The \emph{paired-comparison permutation null distribution} assigns a probability of \(2^{-n}\) to each of the \(2^n\) possible configurations, where for each \(i\), \(X_i\) is assigned to sample 1 and \(Y_i\) to sample 2, or \(Y_i\) is assigned to sample 1 and \(X_i\) to sample 2.
Throughout this paper, unless otherwise specified, we use $\bP$, $\bE$, $\var$, and $\cov$ to represent probability, expectation, variance, and covariance, respectively, under this paired-comparison permutation null distribution.

{
We begin by constructing a similarity graph on the pooled observations. Using distance metrics such as the Euclidean distance, a similarity graph can be built based on various criteria.  
For instance, a similarity graph \( G \) can be a minimum spanning tree (MST), which is a spanning tree that connects all observations while minimizing the total distance of the edges. The choice of the distance metric is not restricted to the Euclidean distance; alternatives such as the $L_1$ norm or other norms can also be used. When observations do not lie in Euclidean space, other forms of dissimilarity information may be employed; more discussions on the choice of the distance can be found in \cite{chen2013}. 

Beyond the MST, other methods of constructing the similarity graph can be applied. For example, a nearest neighbor graph connects each observation to its nearest neighbor. Denser graph structures, such as the $k$-MST, can also be considered. The $k$-MST is the union of the 1st, 2nd, $\dots$, the $k$th MST, where the 1st MST corresponds to the standard MST. Each subsequent $l$th MST is a spanning tree that connects all observations while minimizing the total distance of the edges, under the condition that it does not reuse any edges from the previous 1st, $\dots$, $(l-1)$th MSTs. 
%The use of similarity graphs has proven effective for two-sample hypothesis testing problems involving high-dimensional data \citep[e.g.,][]{friedman1979, chen2017new}. 

Fig.\ \ref{fig:graph} illustrates examples of similarity graphs constructed using the MST with the Euclidean distance under the null and alternative hypotheses, highlighting how the graph's characteristics vary across different scenarios.
In the figure, circular and square nodes represent observations from samples 1 and 2, respectively, while
 nodes labeled \( i \) and \( i^* \) indicate observations from the same pair.  
Under the null hypothesis (Fig.\ \ref{fig:graph}(a)), circular and square nodes are typically well mixed, with a substantial number of edges connecting nodes from different samples. 
Under a mean difference scenario (Fig.\ \ref{fig:graph}(b)), nodes tend to connect more frequently to others from the same sample, reflecting group separation.
In the case of a variance difference (Fig.\ \ref{fig:graph}(c)), circular nodes with smaller variance predominantly form within-sample connections, while square nodes with larger variance are more likely to connect to nodes with smaller variance. This behavior is attributed to the curse of dimensionality, where the volume increases exponentially with dimension. Observations from the larger variance distribution are sparsely scattered, causing those nearer-to-the-center observations to appear more proximate. In contrast, observations from the smaller variance distribution are relatively concentrated around the center of the data.
}
\begin{figure}[!h]
  \centering
  \includegraphics[width=\textwidth]{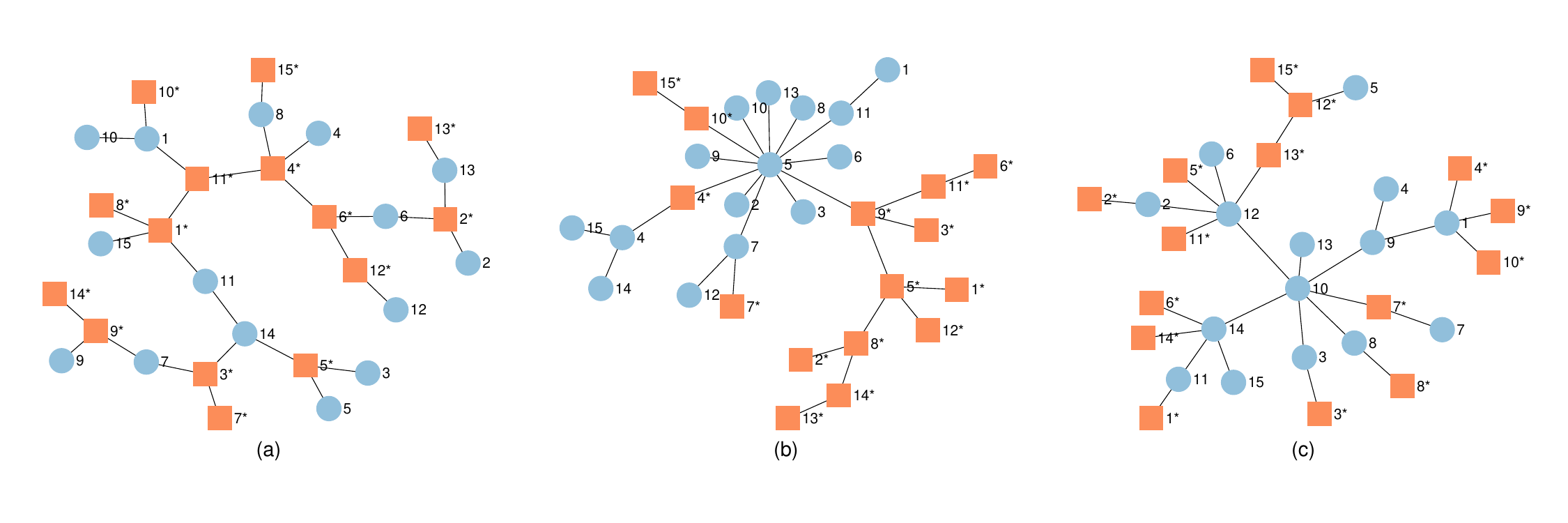}
  \vspace{-.65cm}
  \caption{Examples of similarity graphs $G$ constructed on 100-dimensional data: (a) under the null hypothesis; (b) with a mean difference; (c) with a variance difference. Circular and square nodes represent observations from different samples.}\label{fig:graph}
\end{figure}

Since the graph is constructed on the pooled observations, it remains unchanged under the paired-comparison permutation null distribution. 
For simplicity of notation, we do not include \( G \) as a subscript for quantities that depend on \( G \).  
Let \( R_1 \) denote the number of edges in the similarity graph that connect observations within sample 1, and \( R_2 \) denote the number of edges that connect observations within sample 2. 
The new test statistic is defined as:
\begin{equation}\label{eq:Zg}
D = \begin{pmatrix}
R_1 - \bE(R_1) \\
R_2 - \bE(R_2)
\end{pmatrix}^\top
\mathbf{\Sigma}_R^{-1}
\begin{pmatrix}
R_1 - \bE(R_1) \\
R_2 - \bE(R_2)
\end{pmatrix},
\end{equation}
where \(\mathbf{\Sigma}_R = \var((R_1, R_2)^\top)\). 
The null hypothesis is rejected at the significance level \(\alpha\) if \( D > C(\alpha) \), where $C(\alpha)$ is a critical value.
\begin{remark}
Although the form of \( D \) is identical to the test statistic in \cite{chen2017new}, the expectation and variance in \eqref{eq:Zg} are computed under the paired-comparison permutation null distribution, which accounts for the paired structure of the data. Consequently, their analytical formulas differ from those in \cite{chen2017new}.
\end{remark}

In the following, we derive exact analytical expressions for \(\bE(R_1)\), \(\bE(R_2)\), \(\var(R_1)\), \(\var(R_2)\), and \(\cov(R_1, R_2)\), enabling efficient computation of the proposed test statistic (Section \ref{sec:expression}). We also discuss how to determine the critical value \( C(\alpha) \) analytically (Section \ref{sec:asym}).

\subsection{Analytic expressions for the new paired test statistic}\label{sec:expression}

Let $N=2n$ be the total number of observations and let $Z_i=I(i\leq n)X_i+I(i>n)Y_{i-n}$, $i\in\{1,\ldots, N\}$, where $I(\cdot)$ is the indicator function.
Let $g_i$ be an indicator function that equals 1 when $Z_i$ is assigned to sample 1 under the paired-comparison permutation null distribution, and 0 if $Z_i$ is assigned to sample 2. 
It is straightforward to verify that $\bP(g_i = 1) = 0.5$. 
We use $a\wedge b$ to denote the minimum of $a$ and $b$, and $a\vee b$ to denote the maximum of $a$ and $b$.
{Let $Z_{i^*}$ denote the observation paired with $Z_i$, where $i^* = i+n$ if $i\leq n$, and $i^*=i-n$ if $i>n$.
Define $\ii=(i\wedge i^*, i\vee i^*)$, where $i$ and $i^*$ are the indices of the two observations in pair $\ii$.  
By construction, we always have $g_i+g_{i^*}=1$. 
Similarly, define $\jj=(j\wedge j^*, j\vee j^*)$, where $j$ and $j^*$ are the indices of the two observations in pair $\jj$. }
Since assigning $Z_i$ to sample 1 is independent of assigning $Z_j$ $(\jj\neq \ii)$ to sample 1, $g_i$ and $g_j$ are independent. 

%\allowdisplaybreaks

For an edge in $G$, we denote it by the indices of the nodes it connects. 
By definition, we have
\vspace{-0.2cm}
\begin{align*}
R_1 = \sum_{(i,j)\in G}I(g_i=g_j=1) \text{ ~ and ~ } R_2 = \sum_{(i,j)\in G}I(g_i=g_j=0),
\end{align*}
where we do not distinguish between edge $(i,j)$ and edge $(j,i)$. 
For the graph $G$, let $|G|$ be the number of edges in the graph. 
Let $G_1$ be the subgraph of $G$ that connects observations from different pairs, i.e., $G_1$ consists of edges $\{(i,j)\in G: j\neq i^*\}$.
Let $G_{1,i}$ be the subgraph of $G_1$ that connects to node $i$ and let $|G_{1,i}|$ denote the degree of node $i$ in $G_1$. 
Define $C_1$ as the number of pairs of edges $(i,j),~(i^*,j^*)\in G_1$, and $C_2$ as the number of pairs of edges $(i,j),~(i,j^*)\in G_1$. The analytic expressions are provided in the following theorem.
\begin{theorem}\label{th:expression0}
The analytic expressions of the expectations and variances under the paired-comparison permutation null are as follows:
\vspace{-.1cm}
\begin{flalign*}
& \hspace{1.5cm}\bE(R_1) = \bE(R_2) = \frac{1}{4}|G_1|,& 
\end{flalign*}
\vspace{-.75cm}
\begin{flalign*}
& \hspace{1.5cm}\var(R_1) = \var(R_2) = \frac{1}{16}(|G_1|+2C_1-2C_2)+\frac{1}{16}\sum_{i=1}^n(|G_{1,i}|-|G_{1,i^*}|)^2, &
\end{flalign*}
\vspace{-.75cm}
\begin{flalign*}
& \hspace{1.5cm}\cov(R_1,R_2) = \frac{1}{16}(|G_1|+2C_1-2C_2)-\frac{1}{16}\sum_{i=1}^n(|G_{1,i}|-|G_{1,i^*}|)^2.&
\end{flalign*}
\end{theorem}
\begin{remark}
The analytical expressions for the expectation and variance of $(R_1,R_2)^\top$ differ significantly from those under the permutation null distribution for the two-sample test setting. 
{
Notably, these expressions depend solely on $G_1$, the set of edges connecting observations from different pairs, rather than the entire graph $G$. 
The exclusion of $G\setminus G_1$ (the portion of \(G\) not included in \(G_1\)) can be seen as follows.
First, under the paired-comparison permutation null distribution, the two endpoints of an edge connecting observations from the same pair will always belong to different samples. As a result, the subgraph \(G \setminus G_1\) has no influence on the test statistic.  
Second, the edges in \(G \setminus G_1\) predominantly indicate that differences within paired observations are generally smaller compared to those between non-paired observations. However, these edges provide limited insight into subtle differences within paired observations.}
%Hence, one could construct a similarity graph $G$ that does not include any edge connecting within a pair. 
%It is reasonable to not use the information of edge connecting within pairs as it only reflects that the non-paired observations are distinct enough rather than the paired observations are not changed under different circumstances.
\end{remark}

\begin{proof}
First, notice that for $j\neq i$ and $j\neq i^*$, $g_i$ and $g_j$ are independent. Thus,
\[
\bP(g_i=g_j=1) = \bP(g_i=1)\bP(g_j=1)=\frac{1}{4}.
\]
Additionally, note that $g_i\neq g_{i^*}$ always holds. Therefore, $R_1 = \sum_{(i,j)\in G_1}I(g_i=g_j=1)$ and $R_2 = \sum_{(i,j)\in G_1}I(g_i=g_j=0)$.
Then
$\bE(R_1) = \sum_{(i,j)\in G_1}\bP(g_i=g_j=1) = \frac{1}{4}|G_1|.$ Similarly, $\bE(R_2)=\frac{1}{4}|G_1|$.

To compute $\var(R_1)$, we first determine $\bE(R_1^2)$:
\begin{align}\label{eq:ER1sq}
& \bE(R_1^2) =  \sum_{(i,j)\in G_1}\bP(g_i=g_j=1) + \sum_{\substack{(i,j),(i,u)\in G_1 \\ j\neq u}}\bP(g_i=g_j=g_u=1) \notag\\
& \hspace{15mm} + \sum_{\substack{(i,j),(u,v)\in G_1 \\ i,j,u,v \text{ all different}}}\bP(g_i=g_j=g_u=g_v=1)\notag\\
&\hspace{-3mm} =\frac{1}{4}|G_1|+ \sum_{\substack{(i,j),(i,u)\in G_1 \\ j\neq u}}\bP(g_i=g_j=g_u=1) + \sum_{\substack{(i,j),(u,v)\in G_1 \\ i,j,u,v \text{ all different}}}\bP(g_i=g_j=g_u=g_v=1).
\end{align}
We next figure out $\sum_{\substack{(i,j),(i,u)\in G_1 \\ j\neq u}}\bP(g_i=g_j=g_u=1)$ and $\sum_{\substack{(i,j),(u,v)\in G_1 \\ i,j,u,v \text{ all different}}}\bP(g_i=g_j=g_u=g_v=1)$. For $(i,j),(i,u)\in G_1,j\neq u$, it is clear that $i$ and $j$ are from different pairs, and $i$ and $u$ are from different pairs. Since $j\neq u$, if $j$ and $u$ are from the same index pair, then $g_j\neq g_u$. Hence, 
\begin{align}\label{eq:p1}
&\sum_{\substack{(i,j),(i,u)\in G_1 \\ j\neq u}}\bP(g_i=g_j=g_u=1) =  \sum_{\substack{(i,j),(i,u)\in G_1 \\ j\neq u,j\neq u^*}}\bP(g_i=g_j=g_u=1) 
= \frac{1}{8}\sum_{\substack{(i,j),(i,u)\in G_1 \\ j\neq u}}I(j\neq u^*) \notag \\
&= \frac{1}{8}\sum_{i=1}^N\left\{|G_{1,i}|(|G_{1,i}|-1)-\sum_{j,u\in G_{1,i}}I(j=u^*)\right\}
= \frac{1}{8}\left(\sum_{i=1}^N|G_{1,i}|^2-2|G_1|-2C_2\right).
\end{align}
For $(i,j),(u,v)\in G_1$, $i,j,u,v$ all different, by similar arguments as above, if some of $i,j,u,v$ are from the same index pair, we could not have $g_i=g_j=g_u=g_v$. Hence, 
\begin{align*}
&\sum_{\substack{(i,j),(u,v)\in G_1 \\ i,j,u,v \text{ all different}}} \bP(g_i=g_j=g_u=g_v=1) =  \sum_{\substack{(i,j),(u,v)\in G_1 \\ i,j,u,v \text{ all different index pair}}}\bP(g_i=g_j=g_u=g_v=1)\\
& = \frac{1}{16}\left\{\sum_{\substack{(i,j),(u,v)\in G_1 \\ i,j,u,v \text{ all different}}}1 -\sum_{\substack{(i,j),(u,v)\in G_1 \\ i,j,u,v \text{ all different}}} I(u=i^*,v=j^*)  -\sum_{\substack{(i,j),(u,v)\in G_1 \\ i,j,u,v \text{ all different}}} I(u\neq i^*,v=j^*)\right\}.
\end{align*}
Since
\begin{align*}
& \sum_{\substack{(i,j),(u,v)\in G_1 \\ i,j,u,v \text{ all different}}}1 = |G_1|^2-|G_1|-\sum_{i=1}^N|G_{1,i}|(|G_{1,i}|-1), 
\sum_{\substack{(i,j),(u,v)\in G_1 \\ i,j,u,v \text{ all different}}} I(u=i^*,v=j^*)=2C_1, \\
& \sum_{\substack{(i,j),(u,v)\in G_1 \\ i,j,u,v \text{ all different}}} I(u\neq i^*,v=j^*) \\
& = \sum_{(i,j)\in G_1}(|G_{1,i^*}|+|G_{1,j^*}|-I(i\in G_{1,j^*})-I(j\in G_{1,i^*}) -2I(i^*\in G_{1,j^*}))\\
& =  2\sum_{i=1}^n|G_{1,i}||G_{1,i^*}|-2C_2-4C_1,
\end{align*}
we have 
\begin{align}\label{eq:p2}
& \sum_{\substack{(i,j),(u,v)\in G_1 \\ i,j,u,v \text{ all different}}} \bP(g_i=g_j=g_u=g_v=1) \notag \\
= & \frac{1}{16}\left(|G_1|^2+|G_1|-\sum_{i=1}^N|G_{1,i}|^2-2\sum_{i=1}^n|G_{1,i}||G_{1,i^*}| + 2C_1+2C_2\right).
\end{align}
So plugging (\ref{eq:p1}) and (\ref{eq:p2}) into (\ref{eq:ER1sq}), we obtain
\[
\bE(R_1^2)=
\frac{1}{16}\left(|G_1|^2+|G_1|+\sum_{i=1}^N|G_{1,i}|^2-2\sum_{i=1}^n|G_{1,i}||G_{1,i^*}|  + 2C_1-2C_2\right).
\]

\allowdisplaybreaks
For $\cov(R_1,R_2)$, we only need to figure out $\bE(R_1R_2)$. We have 
\begin{align*}
\bE(R_1R_2) = & \sum_{\substack{(i,j),(u,v)\in G_1 \\ i,j,u,v \text{ all different}}}\bP(g_i=g_j=1,g_u=g_v=0) \\
= & \sum_{\substack{(i,j),(u,v)\in G_1 \\ i,j,u,v \text{ all different index pair}}}\bP(g_i=g_j=1,g_u=g_v=0) \\
& + \sum_{\substack{(i,j),(u,v)\in G_1 \\ i,j,u,v \text{ belong to three different index pairs}}}\bP(g_i=g_j=1,g_u=g_v=0) \\
& +\sum_{\substack{(i,j),(u,v)\in G_1 \\ i,j,u,v \text{ belong to two different index pairs}}}\bP(g_i=g_j=1,g_u=g_v=0) \\
= & \frac{1}{16}\left(|G_1|^2+|G_1|-\sum_{i=1}^N|G_{1,i}|^2-2\sum_{i=1}^n|G_{1,i}||G_{1,i^*}|  + 2C_1+2C_2\right)\\
& + \frac{1}{8}\left(2\sum_{i=1}^n|G_{1,i}||G_{1,i^*}|-2C_2-4C_1\right)+\frac{2}{4}C_1\\
= & \frac{1}{16}\left(|G_1|^2+|G_1|-\sum_{i=1}^N|G_{1,i}|^2+2\sum_{i=1}^n|G_{1,i}||G_{1,i^*}|  +2C_1-2C_2\right).
\end{align*}
\end{proof}

%\begin{corollary}\label{cor:expression}
%The expectation and variance of $R_1+R_2$ and $R_1-R_2$ under the paired-comparison permutation null are:
%\begin{align*}
%& \bE(R_1+R_2) = \frac{1}{2}|G_1|,\\
%& \var(R_1+R_2) = \frac{1}{4}(|G_1|+2C_1-2C_2),\\
%& \bE(R_1-R_2) = 0,\\
%& \var(R_1-R_2) = \frac{1}{4}\sum_{i=1}^n(|G_{1,i}|-|G_{1,i^*}|)^2.
%\end{align*}
%\end{corollary}
%Corollary \ref{cor:expression} follows straightforwardly from Theorem \ref{th:expression0}.

To ensure that the proposed test statistic $D$ is well defined, $\Sigma_R$ needs to be invertible.
\begin{theorem}\label{th:graph}
The proposed test statistic $D$ is well defined except the following two scenarios: 
\begin{enumerate}
\item For each pair $\ii$, the two nodes have the same degree in $G_1$, i.e., $|G_{1,i}|-|G_{1,i^*}|= 0$ for all $i$; 
\item $|G_1|+2C_1-2C_2 =0$.
\end{enumerate}
\end{theorem}
\begin{remark}
This theorem follows directly from the analytic expression of $\Sigma_R$ derived in Theorem \ref{th:expression0}. After  simplifications, the determinant of $\Sigma_R$ is given by
\[
|\Sigma_R|=\frac{1}{64}(|G_1|+2C_1-2C_2)\sum_{i=1}^n(|G_{1,i}|-|G_{1,i^*}|)^2.
\]
Hence, $D$ is well defined except when $|\Sigma_R|=0$, which occurs in the two scenarios outlined in Theorem \ref{th:graph}.

Let $subG_1^{\ii,\jj}$ be the subgraph of $G_1$ that connects any nodes in the pairs $\ii$ and $\jj$. 
If there is at least one edge in $subG_1^{\ii,\jj}$, then it has eight possible configurations, as shown in Fig.\ \ref{paired}. 

\begin{figure}[!h]
  \centering
  \includegraphics[width=0.95\textwidth]{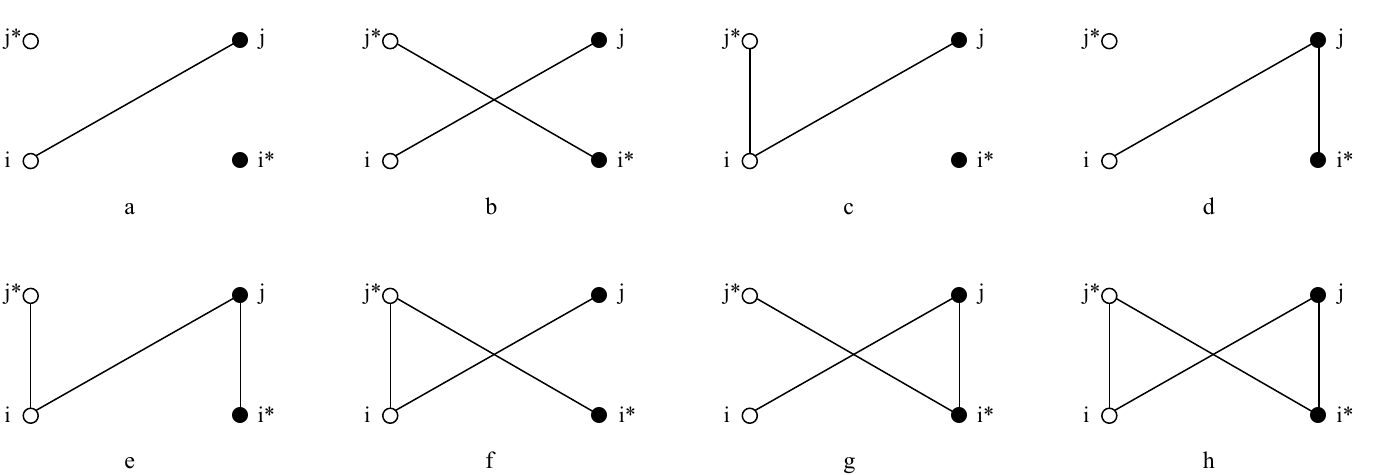}
  \caption{For $(i,j)\in G_1$, eight possible subgraphs of $G_1$ between pairs $\ii$ and $\jj$.}\label{paired}
\end{figure}

Notice that
\[
|G_1|+2C_1-2C_2 = \sum_{\substack{\text{pairs }\ii, \jj \\ |subG_1^{\ii,\jj}|>0}} T(subG_1^{\ii,\jj}),
\]
where
\begin{align*}
T(subG_1^{\ii,\jj}) := &|subG_1^{\ii,\jj}| + 2\times\text{pairs of two edges in }subG_1^{\ii,\jj} \text{ not sharing any node } \\
& - 2\times \text{pairs of two edges in }subG_1^{\ii\jj} \text{ sharing a node}.
\end{align*}
It is straightforward to verify that
\[
T(c) = T(d) = T(h) = 0,\quad T(a) = T(e) = T(f) = T(g) = 1,\quad T(b) = 4.
\]
Hence, if every subgraph $subG_1^{\ii,\jj}$ belongs to one of the three configurations $c$, $d$, or $h$, then $\Sigma_R$ is noninvertible.
\end{remark}

\subsection{Asymptotics}\label{sec:asym}
For the critical value $C(\alpha)$, it can be determined by performing the paired-comparison permutation directly.
However, this approach is time-consuming. 
To make the test more practical and application-friendly, we study the asymptotic distribution of the statistic $D$. 

Before stating the theorem, we define two additional terms on the similarity graph $G_1$:
For an edge $e\in G_1$, let $e_-$ and $e_+$ be the indices of the nodes connected by the edge $e$.
\begin{align*}
& A_e = \{(i,j)\in G_1:i\in \{e_-,e_+,e_-^*,e_+^*\} \text{ or } j\in\{e_-,e_+,e_-^*,e_+^*\}\}, \\
& B_e = \cup_{\tilde e\in A_e}A_{\tilde e}.
\end{align*}
We use $a = O(b)$ to denote that $a$ and $b$ are of the same order, and $a = o(b)$ to denote that $a$ is of a smaller order than $b$. 

To derive the asymptotic behavior of our statistic, we work under the following conditions for some $\gamma>0$:
\begin{condition}\label{cond1}
$\sum_{e\in G_1}|A_e||B_e|=o(N^{1.5\gamma})$.
\end{condition}
\begin{condition}\label{cond2}
$\sum_{i=1}^n(|G_{1,i}|-|G_{1,i^*}|)^2=O(N^\gamma)$.
\end{condition}
\begin{condition}\label{cond3}
$|G_1|+2C_1-2C_2=O(N^\gamma)$.
\end{condition}
\begin{remark}\label{rem:conds}
The parameter $\gamma$ can be any positive number. For example, when $\gamma=1$, the conditions above simplify to the following conditions (\ref{cond:1}), (\ref{cond:2}), and (\ref{cond:3}), respectively.
\begin{equation}\label{cond:1}
\sum_{e\in G_1}|A_e||B_e|=o(N^{1.5}),
\end{equation}
\begin{equation}\label{cond:2}
\sum_{i=1}^n(|G_{1,i}|-|G_{1,i^*}|)^2=O(N),
\end{equation}
\begin{equation}\label{cond:3}
|G_1|+2C_1-2C_2=O(N).
\end{equation}
Here, Condition (\ref{cond:1}) imposes a constraint on the number of edges sharing a pair in the graph $G_1$, ensuring that this number is not too large. A similar condition was proposed for graph-based statistics for independent observations and discussed in \citet{chen2017new} and \citet{chen2018weighted}.

Conditions (\ref{cond:2}) and (\ref{cond:3}) ensure that $(R_1,R_2)$ does not degenerate asymptotically.
Let 
\[
L_1 = \{\ii=(i\wedge i^*,i\vee i^*): |G_{1,i}|\neq |G_{1,i^*}|\}.
\]
If $|L_1|=O(N)$ and $(|G_{1,i}|-|G_{1,i^*}|)^2=O(1),~\ii\in L_1$, then (\ref{cond:2}) is satisfied.

Condition (\ref{cond:3}) places a constraint on the structure of the graph $G_1$. As shown in the proof of Theorem \ref{th:graph}, we obtain
\[
|G_1|+2C_1-2C_2 = \sum_{\substack{\text{pairs }\ii, \jj \\ |subG_1^{\ii,\jj}|>0}} T(subG_1^{\ii,\jj}),
\]
where $T(subG_1^{\ii,\jj})=0,1$, or $4$. Let 
\[
L_2 = \{subG_1^{\ii,\jj}: subG_1^{\ii,\jj}\text{ contains at least one edge, i.e., }|subG_1^{\ii,\jj}|>0\},
\]
\[
L_3 = \{subG_1^{\ii,\jj}\in L_2: T(subG_1^{\ii,\jj})\neq 0\}.
\]
We have 
\[
|G_1|/4\leq|L_2|\leq |G_1|.
\]
Thus, $|L_2|=O(|G_1|)$. If $|G_1|=O(N)$ and $|L_3|=O(|L_2|)$, then (\ref{cond:3}) is satisfied.
\end{remark}

\begin{theorem}\label{th:asym}
 Under Conditions \ref{cond1}, \ref{cond2}, and \ref{cond3}, as $N\rightarrow\infty$, $((R_1-\bE(R_1))/\sqrt{\var(R_1)}$, $(R_2-\bE(R_2))/\sqrt{\var(R_2)})^\top$ converges in distribution to a bivariate Gaussian distribution under the paired-comparison permutation null distribution.
\end{theorem}
The proof of Theorem \ref{th:asym} is provided in \ref{pf:asym}. 
Based on Theorem \ref{th:asym}, the asymptotic distribution of $D$ can be easily derived.
\begin{corollary}\label{corollary}
Under Conditions \ref{cond1}, \ref{cond2}, and \ref{cond3}, as $N\rightarrow\infty$, 
\[
\quad D\longrightarrow\chi_2^2
\]
in distribution under the paired-comparison permutation null distribution.
\end{corollary}

We reject the null hypothesis at the significance level $\alpha$ when $D > C(\alpha)$. Based on Corollary~\ref{corollary}, $C(\alpha)$ can be approximated by $\chi_2^2(1-\alpha)$, the $(1-\alpha)$-quantile of the $\chi_2^2$ distribution.

{ 
To evaluate how well the $\chi^2_2$ distribution approximates the finite-sample distribution of $D$, we compare the two using quantile-quantile (QQ) plots. Specifically, we consider the following three data-generating settings:
\begin{itemize}
  \item Setting 1 (S1): $X_i \stackrel{\text{iid}}{\sim} \mathcal{N}_d(\mathbf{0}_d, \mathbf{I}_d)$;
  \item Setting 2 (S2): $X_i \stackrel{\text{iid}}{\sim} \text{multivariate } t_3(\mathbf{0}_d, \mathbf{I}_d)$;
  \item Setting 3 (S3): $X_i \stackrel{\text{iid}}{\sim}$ multivariate Laplace distribution with mean $\mathbf{0}_d$ and covariance $\mathbf{I}_d$.
\end{itemize}

For each setting, we generate data for 1000 subjects, randomly assigning $n$ subjects to treatment and the remainder to control. We then match $n$ control subjects to the treated subjects using their propensity scores. We consider $n = 50$ and dimensions $d = 50$ and $d = 100$. For each configuration, 1000 simulation runs are conducted.
Here, we use the 5-MST based on Euclidean distance as the similarity graph for computing the statistic $D$. 

Fig.\ \ref{fig:qqplot} displays QQ plots comparing the sample quantiles of $D$ to the theoretical $\chi^2_2$ quantiles, along with 95\% confidence bands (computing using the R function \texttt{stat\_qq\_band} in \texttt{ggplot2}). The plots show that the sample quantiles of $D$  closely follow those of $\chi^2_2$, with slight deviations in the upper tail. In all scenarios, the correlations between the sample and theoretical quantiles exceed 0.996.

\begin{figure}[!h]
\centering
\includegraphics[width=\textwidth]{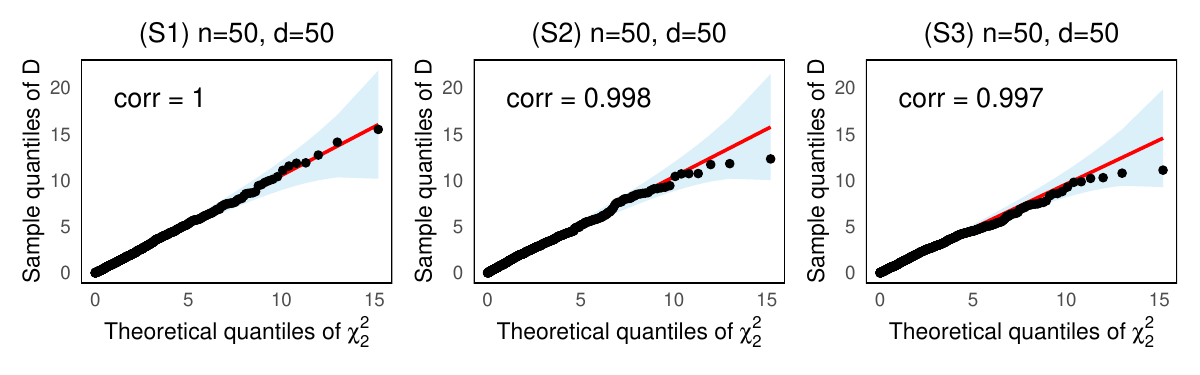}
\includegraphics[width=\textwidth]{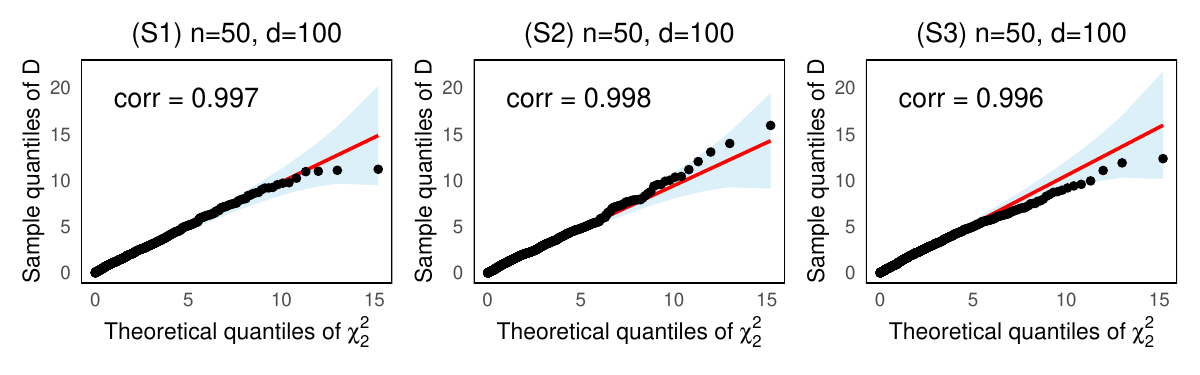}
\caption{{Quantile-quantile plots comparing the empirical distribution of $D$ with the theoretical $\chi^2_2$ distribution under different simulation settings. 
%Each panel corresponds to one setting, with $n=50$ matched pairs and dimensions $d = 50$ or $100$. 
The red line represents the ideal reference line, and the shaded blue area indicates the 95\% confidence band.} } %The sample quantiles of $D$ closely follow the $\chi^2_2$ distribution across all settings.
%The correlation between the sample and theoretical quantiles is displayed in the top left corner of each panel, with values exceeding 0.996 in all cases.}
\label{fig:qqplot}
\end{figure}

We next evaluate how well the rejection rule $D > \chi_2^2(1-\alpha)$ controls the Type I error rate. Using the same three data-generating settings, we vary the number of matched pairs $n = 50,\, 100,\, 150$ and the dimensions $d = 10,\, 50,\, 100$. For each configuration, we conduct 1000 simulation runs.
Table~\ref{tab:check1} reports the empirical Type I error rates under nominal significance levels of 0.05 and 0.1, where the empirical size is defined as the proportion of simulations in which $D > \chi_2^2(1-\alpha)$. The results show that Type I error is well controlled across all settings, even with relatively small sample sizes.

Together, the QQ plots and Type I error evaluations demonstrate the accuracy of the asymptotic $\chi_2^2$ approximation. However, when sample sizes are very small, the approximation may become less reliable. 
In such cases, paired-comparison permutations can be used to generate an empirical null distribution for $D$ and provide more accurate $p$-value.
}

\begin{table}[!h]
\centering
\caption{
Empirical sizes for data generated under different distributions: multivariate normal (S1), multivariate $t$ (S2), and multivariate Laplace (S3), based on 1000 simulation runs
}\label{tab:check1} 
(a) Empirical size at the 0.05 nominal level 
\smallskip

 \begin{tabular}{lccccccccc}\hline\smallskip
   & \multicolumn{3}{c}{$d=10$} & \multicolumn{3}{c}{$d=50$} & \multicolumn{3}{c}{{$d=100$}}\\ 
   & $n=50$ & $n=100$ & $n=150$ & $n=50$ & $n=100$ & $n=150$ & $n=50$ & $n=100$ & $n=150$ \\ 
   (S1) & 0.045 & 0.050 & 0.045 & 0.057 & 0.050 & 0.053 & 0.058 &0.058 & 0.054\\
   (S2) & 0.054 & 0.049 & 0.048 & 0.049 & 0.053 & 0.035 & 0.043 &0.033 & 0.044\\ 
   (S3) & 0.043 & 0.046 & 0.046 & 0.034 & 0.041 & 0.053 & 0.045 & 0.043 & 0.046\\ \hline
  \end{tabular} 

\vspace{0.3cm} 
(b) Empirical size at the 0.1 nominal level 
\smallskip

 \begin{tabular}{lccccccccc}\hline\smallskip
   & \multicolumn{3}{c}{$d=10$} & \multicolumn{3}{c}{$d=50$} & \multicolumn{3}{c}{{$d=100$}} \\ 
   & $n=50$ & $n=100$ & $n=150$ & $n=50$ & $n=100$ & $n=150$ & $n=50$ & $n=100$ & $n=150$\\ 
  (S1) & 0.100 & 0.101 & 0.100 & 0.114 & 0.091 & 0.105 & 0.104 &0.110 & 0.100\\ 
  (S2) & 0.096 & 0.111 & 0.109 & 0.098 & 0.108 & 0.096 & 0.097 &0.094 &0.096\\ 
  (S3) & 0.106 & 0.116 & 0.099 & 0.078 & 0.095 & 0.105 & 0.106 &0.096 &0.105\\ \hline
  \end{tabular} 
\end{table}

\begin{theorem}\label{th:consistency}
If $X$ and $Y$ are independently drawn from two continuous multivariate distributions, the graph $G$ is a $k$-MST $(k = O(1))$ based on the Euclidean distance, and Conditions \ref{cond2} and \ref{cond3} are satisfied with $\gamma=1$, then the test based on $D$ is consistent against all alternatives in the
usual limiting regime.
\end{theorem}
The proof of Theorem \ref{th:consistency} is provided in \ref{pf:consistency}.

%\clearpage

\section{Performance of the proposed test}\label{sec:simu}
In this section, we evaluate the performance of the proposed test $D$ by assessing covariate balance in pair matching and conducting two-sample testing for non-independent paired data. 
{We 
consider the 1-MST, 5-MST, and 10-MST constructed using the Euclidean distance as similarity graphs and denote the corresponding tests as $D1$, $D5$, and $D10$, respectively.
As a sensitivity analysis, we also apply the proposed test using the 1-MST, 5-MST, and 10-MST constructed with the $L_1$ distance. 
The results, which are similar, are provided in \ref{app:addsimu}.}

The proposed test is compared with the multivariate paired Hotelling's $T^2$ test (pHT) and four existing tests for matching: the method of combined differences \citep{hansen2008covariate}, the crossmatch test \citep{rosenbaum2005}, CrossNN, and CrossMST tests \citep{chen2022new}, as well as the Bonferroni-corrected paired $t$-test (BCT) in Section \ref{sec:simu-pairmatch}. 
It is also compared with pHT and BCT in Section \ref{sec:simu-correlated}. 
The significance levels of all tests are set to 0.05. 
For the proposed test $D$, the null hypothesis is rejected when $D>\chi_2^2(0.95)$. 
For the paired Hotelling's $T^2$ test, let $T_i=X_i-Y_i,~i\in\{1,\ldots,n\}$, $\bar T=\sum_{i=1}^nT_i/n$, and $\Sigma_T = \sum_{i=1}^n(T_i-\bar T)(T_i-\bar T)^\top/(n-1)$. The null hypothesis is rejected if $(n-d)n\bar T^\top\Sigma_T^{-1}\bar T/[d(n-1)]>F_{d,n-d}(0.95)$,
where $F_{d,n-d}(0.95)$ denotes the 0.95 quantile of an F-distribution with $d$ and $n-d$ degrees of freedom. 

\subsection{Assessing covariate balance for pair matching}\label{sec:simu-pairmatch}
Following a simulation setting similar to that in \citet{franklin2014metrics}, we assume that $d$-dimensional covariates $\bX_{(1)}=(X_1,\ldots,X_d)^\top$ are observed, where $X_j,~j\in\{1,\ldots,d\}$ are independent and identically distributed according to a Laplace distribution with mean 0 and variance 0.65.
The treatment variable $T$ depends on these $d$ covariates as well as three unobserved transformations of them, $\bX_{(2)}=(X_{d+1},X_{d+2},X_{d+3})^\top=(\sin(X_1)/4,\cos(X_2)/4,X_3^4+X_3^5/5)^\top$. 
We simulate $T$ as a binary variable using the logistic model 
$\text{logit}\{P(T=1)\}=\alpha_0+\balpha_1^\top\bX_{(1)}+\balpha_2^\top\bX_{(2)}$, 
{
where $\alpha_0$ determines the baseline tendency for treatment assignment, and $\balpha_1$ and $\balpha_2$ represent the effects of $\bX_{(1)}$ and $\bX_{(2)}$, respectively, on the log-odds ratio of treatments to controls in the pre-matched data set. 
}
We generate 1000 subjects and determine whether each is treated $(T=1)$ or not $(T=0)$. 
Control subjects $(T=0)$ are then matched to treated subjects $(T=1)$ using propensity scores derived from $\bX_{(1)}$.

\begin{table}[!h]
\centering
\caption{Parameter values for four simulation scenarios. (A) $d=5$; (B) $d=20$} \label{tab:simupara}
\begin{tabular}{l}
      \hline
\multicolumn{1}{c}{(A) Low-dimensional settings ($d=5$)}\\
A1: Zero coefficients for both $\bX_{(1)}$ and $\bX_{(2)}$. \\
\qquad {$\alpha_0=-2.2$}, $\balpha_1=\bzero_{5}$ and $\balpha_2=\bzero_3$. \\
A2: Nonzero coefficients for observed covariates $\bX_{(1)}$ only. \\
\qquad {$\alpha_0=-2.2$}, $\balpha_1=0.2\bone_{5}$ and $\balpha_2=\bzero_{3}$. \\
A3: Nonzero coefficients for unobserved covariates $\bX_{(2)}$ only. \\
\qquad {$\alpha_0=-3$}, $\balpha_1=\bzero_{5}$ and $\balpha_2=0.3\bone_{3}$. \\
A4: Nonzero coefficients for both $\bX_{(1)}$ and $\bX_{(2)}$. \\
\qquad {$\alpha_0=-3$}, $\balpha_1=0.2\bone_{5}$ and $\balpha_2=0.3\bone_{3}$. \\
\hline
\multicolumn{1}{c}{(B) {Moderate-dimensional settings} ($d=20$)}\\
B1: Zero coefficients for both $\bX_{(1)}$ and $\bX_{(2)}$. \\
\qquad {$\alpha_0=-2.2$}, $\balpha_1=\bzero_{20}$ and $\balpha_2=\bzero_3$. \\
B2: Nonzero coefficients for observed covariates $\bX_{(1)}$ only. \\
\qquad {$\alpha_0=-2.4$}, $\balpha_1=0.2\bone_{20}$ and $\balpha_2=\bzero_{3}$. \\
B3: Nonzero coefficients for unobserved covariates $\bX_{(2)}$ only. \\
\qquad {$\alpha_0=-3.2$}, $\balpha_1=\bzero_{20}$ and $\balpha_2=0.45\bone_{3}$. \\
B4: Nonzero coefficients for both $\bX_{(1)}$ and $\bX_{(2)}$. \\
\qquad {$\alpha_0=-3.5$}, $\balpha_1=0.2\bone_{20}$ and $\balpha_2=0.45\bone_{3}$. \\\hline
\end{tabular}
\end{table}

We consider both low-dimensional settings with $d=5$ and {moderate-dimensional settings} with $d=20$. 
Table \ref{tab:simupara} presents the parameter values for various scenarios.
{ The values of $\alpha_0$ are chosen to ensure that the number of pairs remains approximately 100 across all settings.}
For each scenario, we simulate 1000 data sets. 
Scenarios A1 and B1 are used to examine the empirical size, as the covariates of treated subjects and controls are generated from the same distribution.
The standardized mean difference is defined as $\text{SD1}=(\bar x_1-\bar x_0)/\sqrt{(s_1^2+s_0^2)/2}$, where $\bar x_m$ and $s_m^2$ are the sample mean and variance for treated subjects $(m=1)$ and controls $(m=0)$. 
For each covariate, we compute its standardized mean difference between the treatment and control groups, both before and after matching. 
Figs.\ \ref{fig:simu-match1} and \ref{fig:simu-match2} show the boxplots of standardized mean differences for 1000 data sets. The left panels represent differences before matching, while the right panels show differences after matching.
The results indicate that the standardized mean differences for the observed covariates $\bX_{(1)}$ are relatively close to 0 after matching under all scenarios. However, the unobserved covariates $\bX_{(2)}$ remain significantly unbalanced after matching in scenarios where $\balpha_2\neq \bzero_3$ (scenarios A3, A4, B3, B4), suggesting that the distributions of ($X_1,X_2,X_3$) are not well balanced. 
As our objective is to test whether the joint distributions of the covariates in the matched control and treatment groups are identical, we expect a reliable test to reject the null hypothesis in scenarios A3, A4, B3, and B4, while accepting it in scenarios A1, A2, B1, and B2.

\begin{figure}[!h]
\centering
\includegraphics[width=\textwidth]{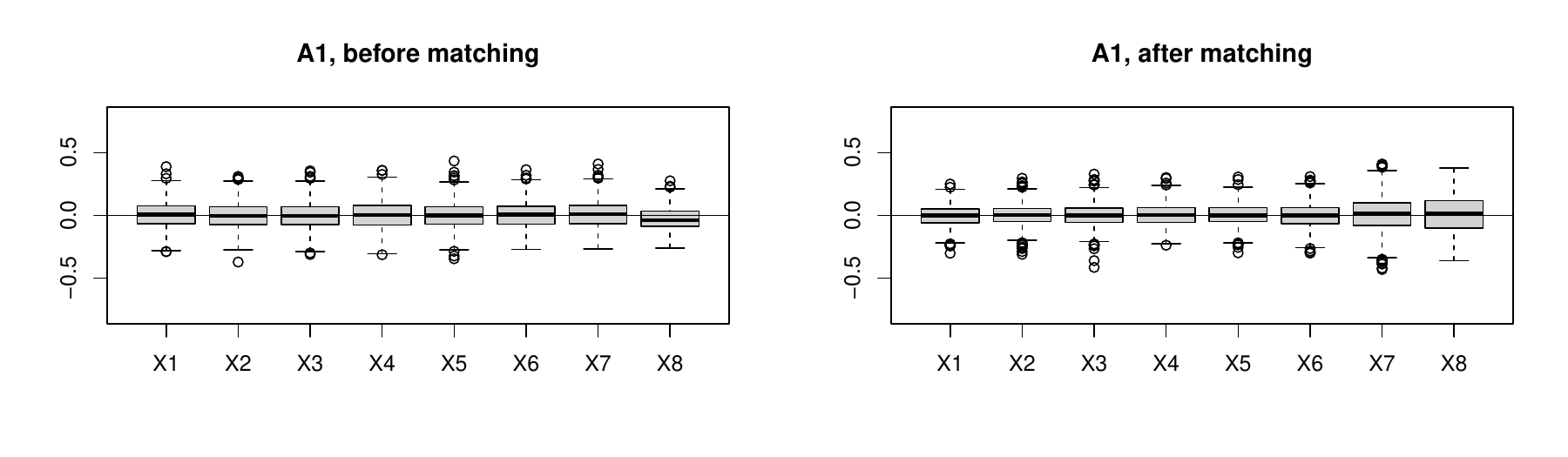}
\includegraphics[width=\textwidth]{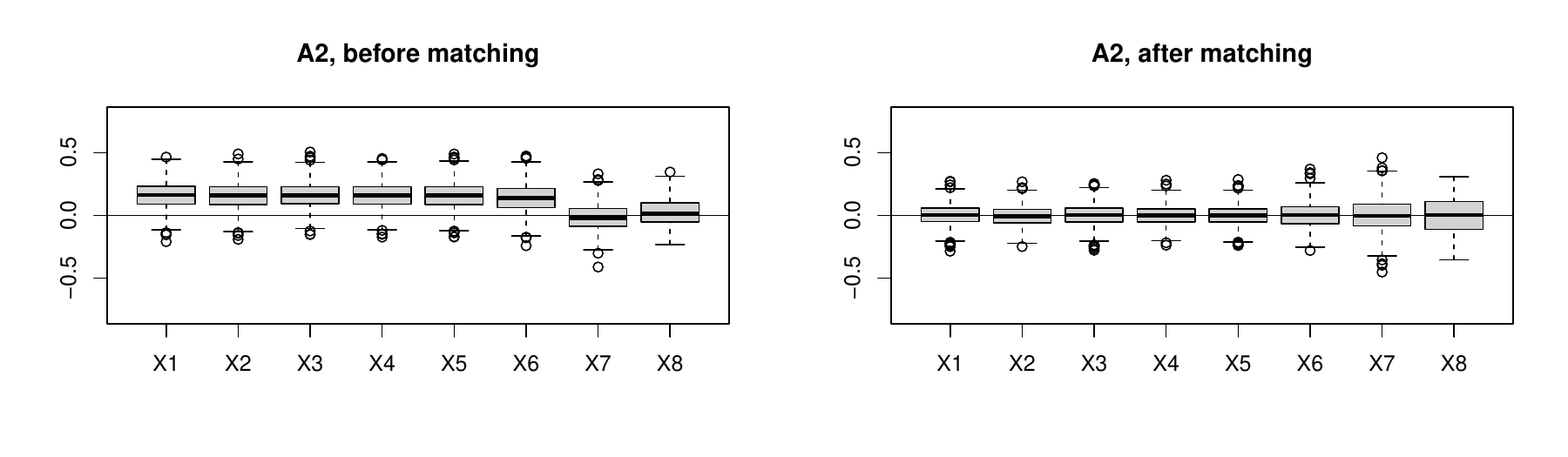}
\includegraphics[width=\textwidth]{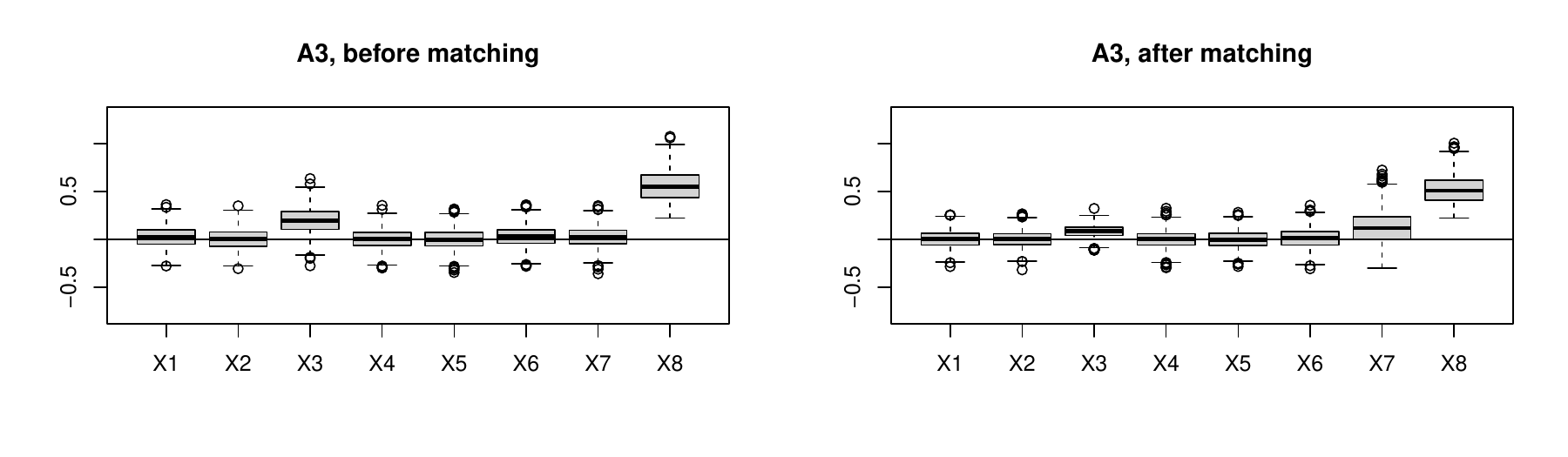}
\includegraphics[width=\textwidth]{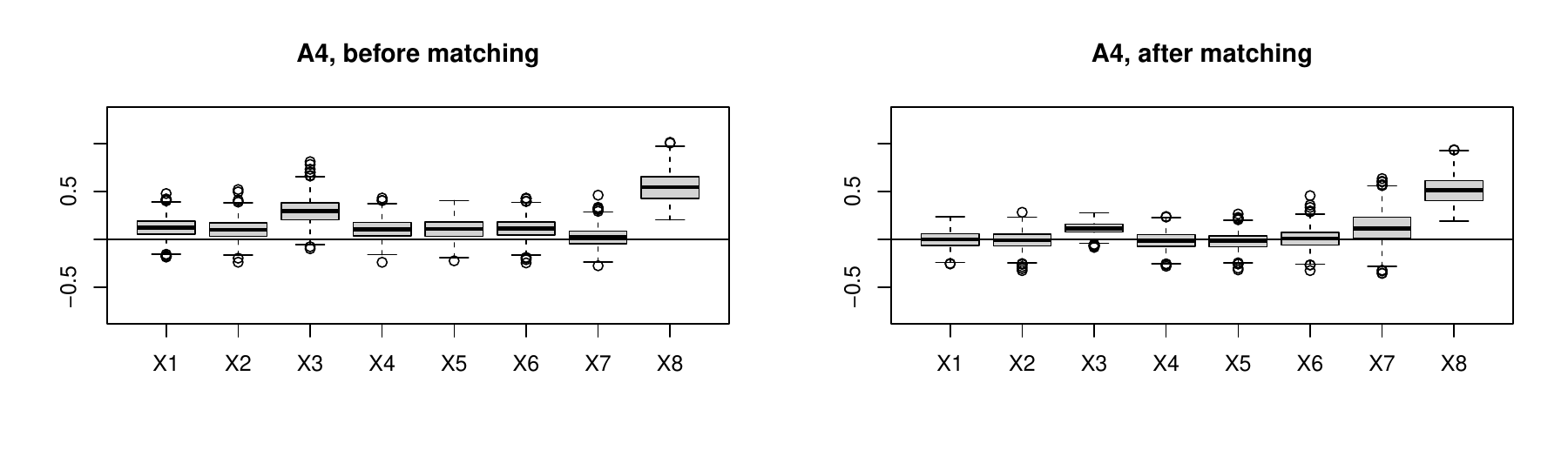}
\caption{Boxplots of standardized mean differences between the treatment and control groups over 1000 runs for each scenario under low-dimensional settings. Left panel: before matching; right panel: after matching.}\label{fig:simu-match1}
\end{figure}

\begin{figure}[!h]
\centering
\includegraphics[width=\textwidth]{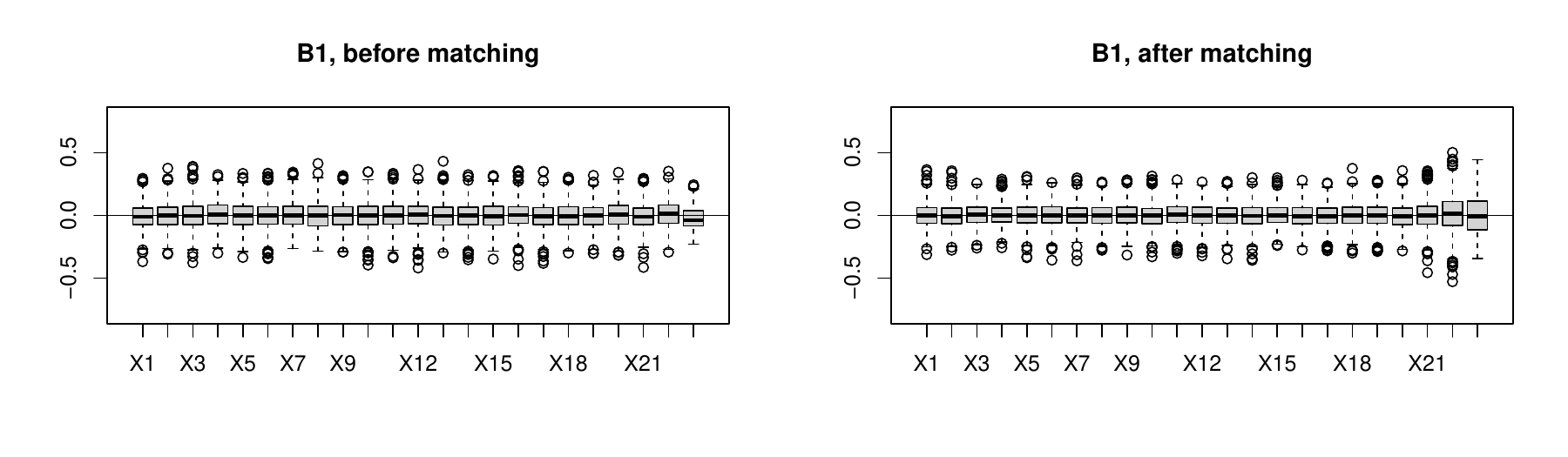}
\includegraphics[width=\textwidth]{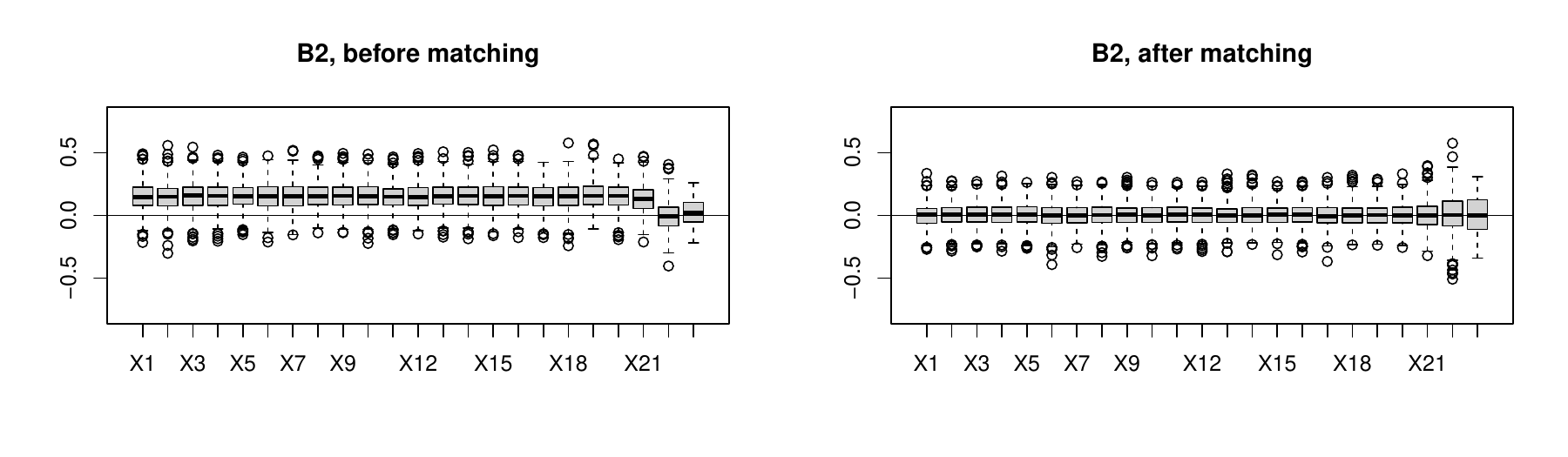}
\includegraphics[width=\textwidth]{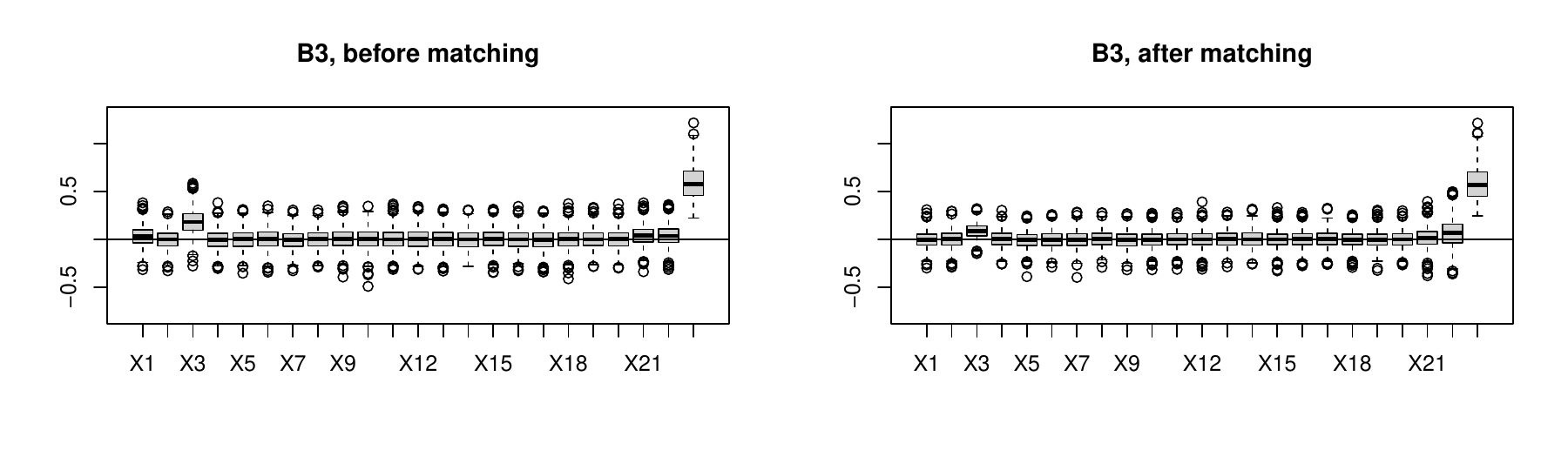}
\includegraphics[width=\textwidth]{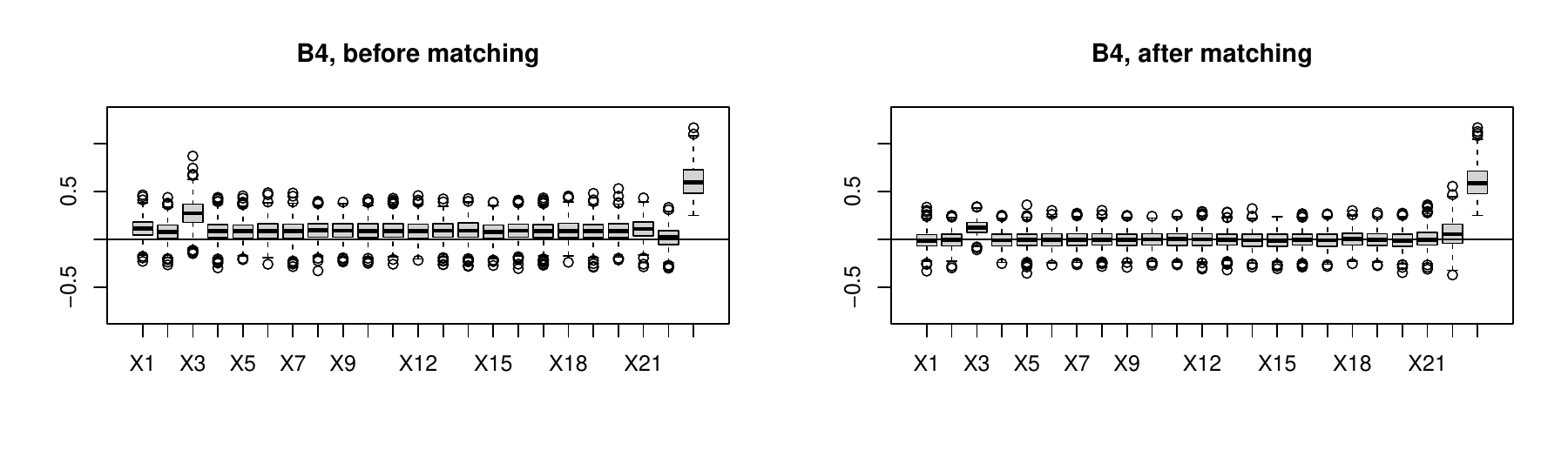}
\caption{Boxplots of standardized mean differences between the treatment and control groups over 1000 runs for each scenario under moderate-dimensional settings. Left panel: before matching; right panel: after matching.}\label{fig:simu-match2}
\end{figure}

\begin{table}[!h]
\centering
\caption{{The average number of pairs (standard deviation in parenthesis)} and the proportion of trials { (out of 1000)} in which the test rejects covariate balance at the 0.05 significance level. The largest estimated power under scenarios A3, A4, B3, and B4 is shown in bold}\label{tab:simu-match}
\begin{tabular}{@{}llccccccccccc@{}}\hline
 & {n} & pHT & CD & CM & CNN1 & CNN5 & CMST1 & CMST5 & BCT & {$D1$} & {$D5$} & {$D10$}\\
& \multicolumn{12}{c}{(A) Low-dimensional settings ($d=5$)} \\
A1 & {100(9)} & 0.002 & 0.001 & 0.026 & 0.020 &0.024 & 0.022 & 0.024 & 0.000 & 0.044 & 0.051 & 0.046 \\
A2 & {104(10)} & 0.000 & 0.000 & 0.029 & 0.022 & 0.014 & 0.022 & 0.021 & 0.000 & 0.040 &0.042 & 0.046 \\ \hline
A3 & {98(10)} & 0.275 & 0.267 & 0.569 & 0.553 & 0.857 & 0.600 & 0.750 & 0.132 & 0.631 &0.899 & \textbf{0.914} \\
A4 & {101(10)} & 0.470 & 0.452 & 0.578 & 0.574 & 0.864 & 0.592 & 0.784 & 0.166 & 0.631 &0.917 & \textbf{0.946} \\ \hline
& \multicolumn{12}{c}{(B)  Moderate-dimensional settings ($d=20$)} \\
B1 & {100(9)} & 0.000 & 0.000 & 0.025 & 0.027 &0.025 & 0.029 & 0.029 & 0.000 & 0.050 & 0.051 & 0.050 \\
B2 & {100(9)} & 0.000 & 0.000 & 0.025 & 0.031 &0.028 & 0.030 & 0.028 & 0.000 & 0.051 & 0.048 & 0.055 \\ \hline
B3 & {104(10)} & 0.047 & 0.031 & 0.517 & 0.439 & 0.559 & 0.394 & 0.429 & 0.017 & 0.536 &0.836 & \textbf{0.875} \\
B4 & {98(9)} & 0.075 & 0.058 & 0.491 & 0.426 &0.543 & 0.382 & 0.407 & 0.016 & 0.535 & 0.839 & \textbf{0.879} \\\hline
\end{tabular}
%\vspace{0.2cm}
\end{table}

We denote the method of combined differences and the crossmatch test as CD and CM, respectively. To apply the CrossNN and CrossMST tests, we consider 1-NN, 5-NN, 1-MST, and 5-MST as the similarity graphs, and denote these tests as CNN1, CNN5, CMST1, and CMST5, respectively.
Here, $k$-NN refers to the $k$-nearest neighbor graph, where two observations $i$ and $j$ are connected by an edge if the distance between $i$ and $j$ is among the $k$-smallest distances from observation $i$ to other observations.
Table \ref{tab:simu-match} presents the average number of pairs (with standard deviation in parentheses) and the proportion of trials in which the tests reject the null hypothesis at the 0.05 significance level. The table shows that the number of pairs is consistently around 100 across all scenarios.
Under the low-dimensional settings, all tests adequately control the empirical size under scenario A1, indicating that the covariates are well-balanced. In the reasonably balanced scenario A2, all tests similarly suggest good covariate balance.
Under scenarios A3 and A4, the proposed tests using denser similarity graphs ($D5$ and $D10$) perform as expected, showing high power. 
Other tests also demonstrate some power but are less effective compared to the proposed methods.
The results under {moderate-dimensional settings} are similar to those in the low-dimensional case, except that pHT, CD, and BCT exhibit no power under scenarios B3 and B4.

\subsection{Two-sample testing for non-independent paired data}\label{sec:simu-correlated}
We first examine the performance of the proposed test for data from the same family of distributions.
We consider the following three settings.
\begin{itemize}
\item Setting 1: $(X_1,Y_1)^\top,(X_2,Y_2)^\top,\ldots,(X_n,Y_n)^\top\stackrel{iid}{\sim}\cN_{2d}(\nu,\Gamma)$;
\item Setting 2:  $(X_1,Y_1)^\top,(X_2,Y_2)^\top,\ldots,(X_n,Y_n)^\top\stackrel{iid}{\sim} \text{multivariate } t_3(\nu,\Gamma)$;
\item Setting 3:  $\ln(X_1,Y_1)^\top,\ln(X_2,Y_2)^\top,\ldots,\ln(X_n,Y_n)^\top\stackrel{iid}{\sim}\cN_{2d}(\nu,\Gamma)$.
\end{itemize}
Here, $\nu=(\nu_1^\top,\nu_2^\top)^\top$ and $\Gamma = \begin{psmallmatrix} \Gamma_1 & \Gamma_{12} \\ \Gamma_{12} & \Gamma_2 \end{psmallmatrix}$.
Let the number of pairs $n$ be fixed at a moderate size with $n=60$. We assess how the proposed statistic behaves when the dimension is comparable to or larger than the number of pairs, by considering three dimensions: $d=50$, $d=100$, and $d=1000$. For each setting, we consider two alternatives for each $d$, as listed below.
\begin{enumerate}
\item[(i)] Only $\nu_1$ differs from $\nu_2$ with $\nu_1=\bzero_d$, $\nu_2=0.5d^{-1/4}\bone_d$, $\Gamma_1=\Gamma_2=\bI_d$ and $\Gamma_{12}=0.6 \bI_d$.
\item[(ii)] Both $\nu_1$ differs from $\nu_2$ and $\Gamma_1$ differs from $\Gamma_2$ with  $\nu_1=\bzero_d$, $\nu_2=0.5d^{-1/4}\bone_d$, $\Gamma_1=\bI_d$, $\Gamma_2=c_d\bI_d$ and $\Gamma_{12}=0.6c_d^{1/2} \bI_d$,
where $c_d = 1.15$ for $d=50$, $c_d=1.1$ for $d=100$, and $c_d=1.05$ for $d=1000$.
\end{enumerate}

\begin{table}[!h]
\centering
\caption{Estimated power at the 0.05 significance level based on 1000 runs. The largest estimated power under each setting is highlighted in bold.}\label{tab:perform}
(a) Data from the same family of distributions

\smallskip

\begin{tabular}{llcccccc}\hline\smallskip
& & \multicolumn{3}{c}{(i) $\nu_1\neq\nu_2$, $\Gamma_1=\Gamma_2$} & \multicolumn{3}{c}{(ii) $\nu_1\neq\nu_2$, $\Gamma_1\neq\Gamma_2$}\\
& $d$ & 50 & 100 & 1000 & 50 & 100 & 1000\\
multivariate normal & pHT & 0.830 & - & - & 0.764 & - & - \\
& BCT & 0.894 & 0.745 & 0.388 & 0.860 & 0.732 & 0.370 \\
& {$D1$} & 0.150 & 0.161 & 0.214 & 0.884 & 0.865 & \textbf{1.000} \\ 
& {$D5$} & 0.706 & 0.746 & 0.800 & 0.968 &0.979 & \textbf{1.000} \\
& {$D10$} & \textbf{0.906} & \textbf{0.912} & \textbf{0.964} & \textbf{0.996} & \textbf{0.995} & \textbf{1.000} \\
multivariate t & pHT & 0.739 & - & - & 0.702 & - & -\\
& BCT & 0.414 & 0.288 & 0.098 & 0.374 & 0.265 & 0.100 \\
& {$D1$} & 0.260 & 0.261 & 0.342 & 0.833 &0.812 & \textbf{0.944} \\ 
& {$D5$} & 0.862 & 0.827 & 0.558 & 0.957 &0.925 & 0.926 \\
& {$D10$} & \textbf{0.907} & \textbf{0.888} & \textbf{0.582} & \textbf{0.981} & \textbf{0.950} & 0.912 \\
multivariate log-normal & pHT & 0.449 & - & - & 0.613 & - & - \\
& BCT & 0.302 & 0.213 & 0.038 & 0.443 & 0.271 & 0.032 \\
& {$D1$} & 0.751 & 0.778 & 0.926 & 0.982 & 0.984 & 0.990\\
& {$D5$} & 0.822 & 0.821 & 0.874 & 0.994 &0.992 & 0.998 \\
& {$D10$} & \textbf{0.888} & \textbf{0.872} & \textbf{0.940} & \textbf{0.997} & \textbf{0.998} & \textbf{1.000} \\ \hline
\end{tabular}

\vspace{0.3cm}
(b) Data from different families of distributions

\smallskip

\begin{tabular}{llcccccc}\hline\smallskip
& & \multicolumn{3}{c}{$\Delta=\bzero_d$} & \multicolumn{3}{c}{$\Delta=0.2\bone_d$}\\
& $d$ & 50 & 100 & 1000 & 50 & 100 & 1000\\ 
$\tau_i\stackrel{iid}{\sim}\text{normal}$ & pHT & 0.040 & - &  - & 0.462 &- & -\\
& BCT & 0.041 & 0.038 & 0.041 & 0.520 & 0.570 & 0.720 \\
& {$D1$} & 0.549 & 0.610 & 0.811 & 0.656 & 0.711 & 0.979 \\
& {$D5$} & 0.616 & 0.753 & 0.972 & 0.770 &0.910 & \textbf{1.000} \\
& {$D10$} & \textbf{0.640} & \textbf{0.789} & \textbf{0.979} & \textbf{0.852} & \textbf{0.950} & \textbf{1.000} \\
$\tau_i\stackrel{iid}{\sim}\text{skew normal}$ & pHT & 0.051 & - &  - & 0.380 &- & -\\
& BCT & 0.051 & 0.064 & 0.049 & 0.370 &0.386 & 0.434 \\
& {$D1$} & 0.504 & 0.520 & 0.782 & 0.507 & 0.629 & 0.954 \\ 
& {$D5$} & 0.554 & 0.667 & 0.966 & 0.629 &0.809 & \textbf{1.000} \\
& {$D10$} & \textbf{0.593} & \textbf{0.707} & \textbf{0.976} & \textbf{0.719} & \textbf{0.886} & \textbf{1.000} \\
$\tau_i\stackrel{iid}{\sim}\text{Laplace}$ & pHT & 0.045 & - &  - & 0.450 &- & -\\
& BCT & 0.040 & 0.032 & 0.024 & 0.524 & 0.557 & 0.714\\
& {$D1$} & 0.322 & 0.412 & 0.691 & 0.416 & 0.510 & 0.947 \\ 
& {$D5$} & 0.349 & 0.490 & 0.911 & 0.589 &0.785 & \textbf{1.000} \\
& {$D10$} & \textbf{0.385} & \textbf{0.532} & \textbf{0.938} & \textbf{0.699} & \textbf{0.907} & \textbf{1.000} \\\hline
\end{tabular}

\end{table}

The results under various scenarios are summarized in Table \ref{tab:perform}(a). We begin by examining the results for moderate-dimensional settings ($d=50$). When only $\nu_1$ differs from $\nu_2$, both pHT and BCT demonstrate high power for the multivariate normal distribution but perform poorly for the multivariate log-normal distribution.
When both $\nu_1\neq\nu_2$ and $\Gamma_1\neq\Gamma_2$, the performance of pHT and BCT is much worse than the proposed tests $D1$, $D5$ and $D10$, particularly for the multivariate $t$-distribution and the multivariate log-normal distribution.
For higher dimensions ($d=100$ and $d=1000$), $D5$ and $D10$ perform well across all settings. In contrast, BCT only exhibits power for the multivariate normal scenario, and paired Hotelling's $T^2$ cannot be applied when $d>n$.

Additionally, under setting (ii), $\nu_1$ and $\nu_2$ are the same as those under setting (i), while $\Gamma_2$ differs from $\Gamma_1$. As a result, the difference between the two samples becomes more pronounced. Comparing the results under setting (ii) with those under setting (i), we observe an increase in the power of $D1$, $D5$, and $D10$ across the three distributions, while the power of pHT and BCT decreases for the multivariate normal and multivariate $t$-distributions.

%\subsection{Two samples from different families of distributions}
Now we examine the performance of the proposed test when the samples $X_i$ and $Y_i$ come from different families of distributions.
Specifically, we generate the data such that
\[
X_i = \alpha_i + \epsilon_i,\quad Y_i = \Delta + \alpha_i + \tau_i,\quad i=1,\ldots,n,
\]
where $\Delta$ is a constant vector, and $\alpha_i,\epsilon_i,\tau_i$ are independent random vectors. 
Here, $\alpha_i$ is drawn from a multivariate normal distribution, $\alpha_i\stackrel{iid}{\sim} \cN_d(\bzero_d,\Omega_1)$, and $\epsilon_i$ follows a multivariate $t$-distribution with 3 degrees of freedom, $\epsilon_i\stackrel{iid}{\sim} t_3(\bzero_d,\Omega_2/3)$. 
We consider three different distributions for $\tau_i$: (i) $\tau_i\stackrel{iid}{\sim}\cN_d(\bzero_d,\Omega_2)$, (ii) $\tau_i\stackrel{iid}{\sim}$ multivariate skew normal distribution with mean $\bzero_d$, variance $\Omega_2$, skewness 1 and (iii) $\tau_i\stackrel{iid}{\sim}$ multivariate Laplace distribution with mean $\bzero_d$, variance $\Omega_2$. Therefore, $\bE(X_i)=\bzero_d$, $\bE(Y_i)=\Delta$, $\var(X_i)=\var(Y_i)=\Omega_1+\Omega_2$ and $\cov(X_i,Y_i)=\Omega_1$. 
We fix the number of pairs at $n=60$ and consider three different dimensions: $d=50,100,1000$. 
Let the covariance matrix $\Omega_1=(\Omega_{ij})$ be defined as $\Omega_{ij}=0.5^{|i-j|}$, and let $\Omega_2=\bI_d$, the identity matrix of dimension $d$.
For each $d$, we consider two settings for $\Delta$: $\Delta=\bzero_d$ and $\Delta=0.2\bone_d$.

Table \ref{tab:perform}(b) presents the results for the case where \(X_i\) and \(Y_i\) are drawn from different families of distributions. We first examine the scenario when there is no mean difference between the two samples (\(\Delta = \bzero_d\)). In this case, the means and variances of \(X_i\) and \(Y_i\) are identical, but the two distributions \(F_X\) and \(F_Y\) differ. A powerful test should reject the null hypothesis, indicating a difference in distributional shapes. 
However, both the pHT and BCT tests show almost no power in this setting, and pHT cannot be applied in high-dimensional scenarios ($d=100$ and $d=1000$). Therefore, they fail to detect the shape difference between the distributions. On the other hand, the proposed tests perform well across all scenarios.
When a mean difference is introduced between the two samples (\(\Delta = 0.2\bone_d\)), pHT exhibits some power under the case $d=50$, but it becomes inapplicable for the high-dimensional scenarios. 
BCT shows some power as well, and the proposed tests \(D5\) and \(D10\) demonstrate consistently high power across all dimensions.

{As seen in the simulation results, the tests \(D5\) and \(D10\) exhibit similar performance, are effective across all scenarios, and are both more powerful than $D1$. Based on these findings, we select \(D5\)--the test using the 5-MST as the similarity graph--for our real application.}

\section{A real application on Alzheimer's disease research}\label{sec:realdata}
In this section, we illustrate the newly developed test in the context of a research project studying Alzheimer's disease. The data were collected by multiple Alzheimer's Disease Centers between September 2005 and December 2018 (see \citet{beekly2007national} for details). Participants provided written informed consent prior to participation. They were evaluated on cognitive performance during their initial visits, followed by approximately annual follow-up visits. The data are recorded in the Uniform Data Set (UDS), a longitudinal, standardized dataset maintained by the National Alzheimer’s Coordinating Center (NACC). More information can be found at \texttt{https://www.alz.washington.edu/}. The dataset consists of approximately 725 variables obtained from nearly 39,412 research volunteers, based on comprehensive evaluations conducted annually as of the December 2018 data freeze. In our study, we selected data from participants with at least four visits. The dataset contains no personally identifiable information.

\subsection{Assessment of covariate balance in pair matching { with respect to sex}}\label{sec:realapp2}
Several studies have previously reported { an association between sex and} Alzheimer's disease \citep{payami1996gender,podcasy2016considering,grimm2016alzheimer}. In this study, we focus on data from the initial visit and consider only white participants who meet the following criteria: no history of stroke, transient ischemic attack, serious heart problems (e.g., atrial fibrillation, cardiac bypass surgery, congestive heart failure), diabetes, brain trauma, Parkinson's disease, seizures, or other neurological or psychiatric disorders; no substance abuse (except alcohol); and no family history of frontotemporal lobar degeneration mutation or Alzheimer's disease mutation, nor any cognitive impairment in first-degree relatives.
In total, we consider 14 covariates (Table \ref{tab:realapp2}(a)). 
{Among these, the covariates \textit{smoking in the past 30 days} (TOB30), \textit{smoking more than 100 cigarettes in life} (TOB100), \textit{history of heart attack} (HATT), and \textit{active depression in the past 2 years} (DEP2YRS) are binary variables, each transformed into a dummy variable.
The covariates \textit{angioplasty, endarterectomy, or stent placement} (ANG) and \textit{alcohol abuse} (ALC) are categorical variables with three levels: absent, recent/active, and remote/inactive.}
These two variables are transformed into two dummy variables, respectively.
We match 105 male participants with 105 female participants from a pool of 153 female participants, aiming to balance the covariates between the female and male groups. 
To achieve this, we use the \texttt{pairmatch()} function from the \texttt{optmatch} package in \textsf{R} \citep{hansen2007optmatch}, with the distances between male and female participants computed using the \texttt{match\_on()} function based on logit propensity scores.
{
\begin{remark}
    { Here, ``sex" does not directly represent a modifiable treatment or intervention. 
    %Instead, the potential outcomes in our study are framed within the context of comparing groups defined by sex, while accounting for potential confounders that may mediate or moderate the association between sex and Alzheimer's disease. 
    Our intention is to explore associations that might stem from biological, social, or environmental differences correlated with sex, rather than to suggest an intervention to alter sex. Understanding that sex is associated with different risks of Alzheimer's disease can inform tailored screening strategies or preventive measures. For example, healthcare providers might prioritize monitoring or early interventions for at-risk groups based on sex-specific risk profiles.}
\end{remark}
}

\begin{sidewaystable}
\centering
\caption{Summary of covariates and statistics before and after pair matching of a male participant and a female participant}\label{tab:realapp2}
(a) Matching results of Section \ref{sec:realapp2}

\smallskip

\addtolength{\tabcolsep}{-1pt}
\begin{tabular}{@{}llrrrrrrrrr@{}}\hline\smallskip
& & \multicolumn{5}{l}{Before matching} & \multicolumn{4}{l}{After matching} \\ 
Description & Covariates & Male & Female & SD1 & SD2 & SD3 & Female & SD1 & SD2 & SD3 \\ 
Age at visit & AGE & 72.762 & 72.092 & 0.081 & -0.130 & 0.048 & 73.057 & -0.037 & 0.027 & -0.119 \\ 
Body Mass Index & BMI & 26.827 & 26.852 & -0.005 & -0.489 & -0.367 & 26.150 & 0.144 & -0.361 & -0.158 \\ 
Years of smoking history & SMOKYRS & 9.114 & 9.248 & -0.010 & -0.238 & -0.516 & 9.819 & -0.053 & -0.120 & -0.090 \\ 
Systolic blood pressure & BPSYS & 135.152 & 133.915 & 0.068 & -0.079 & 0.091 & 134.343 & 0.045 & -0.039 & 0.271 \\ 
Diastolic blood pressure & BPDIAS & 75.505 & 76.307 & -0.084 & -0.329 & -0.069 & 75.105 & 0.043 & -0.263 & 0.108 \\ 
Smoking in the past 30 days & TOB30 & 0.019 & 0.033 & -0.086 & -0.510 & -2.869 & 0.019 & 0.000 & 0.000 & 0.000\\ 
Smoking more than 100 cigarettes in life & TOB100 & 0.486 & 0.425 & 0.122 & 0.022 & -0.238 & 0.514 & -0.057 & 0.000 & 0.113 \\ 
History of heart attack & HATT & 0.114 & 0.013 & 0.422 & 1.534 & 4.741 & 0.019 & 0.387 & 1.364 & 4.037\\ 
Angioplasty, endarterectomy, or stent placement & ANG &  \\
\quad Recent/Active &  \quad ANG1 & 0.019 & 0.000 & 0.196 & 1.981 & 19.620 & 0.000 & 0.196 & 1.981 & 19.620 \\
\quad Remote/Inactive & \quad ANG2 & 0.095 & 0.026 & 0.291 & 1.078 & 3.414 & 0.019 & 0.331 & 1.275 & 4.252 \\ 
Alcohol abuse & ALC & \\
\quad Recent/Active & \quad ALC1 & 0.019 & 0.007 & 0.111 & 0.960 & 8.080 & 0.010 & 0.080 &0.652 & 5.156 \\
\quad Remote/Inactive & \quad ALC2 & 0.029 & 0.020 & 0.058 & 0.360 & 2.112 & 0.029 & 0.000 & 0.000 & 0.000 \\ 
Active depression in the past 2 years & DEP2YRS & 0.229 & 0.294 & -0.149 & -0.162 & 0.120 & 0.229 & 0.000 & 0.000 & 0.000 \\ 
Average number of packs smoked per day & PACKSPER & 1.152 & 1.059 & 0.057 & -0.164 &-1.039 & 1.257 & -0.061 & -0.321 & -1.457 \\ \hline
\end{tabular}
\vspace{1cm}

(b) Matching results of Section \ref{sec:realapp3}

\smallskip

\begin{tabular}{@{}llrrrrrrrrr@{}}\hline\smallskip
& & \multicolumn{5}{l}{Before matching} & \multicolumn{4}{l}{After matching} \\ 
Description & Covariates & Male & Female & SD1 & SD2 & SD3 & Female & SD1 & SD2 & SD3 \\ 
Age at visit & AGE & 74.750 & 73.500 & 0.162 & -0.350 & 0.258 & 74.763 & -0.002 & -0.165 & 0.384 \\ 
Body Mass Index & BMI & 26.429 & 26.894 & -0.095 & -0.698 & -0.816 & 26.101 & 0.074 & -0.407 & -0.775 \\ 
Years of smoking history & SMOKYRS & 9.882 & 10.308 & -0.031 & -0.054 & -0.039 & 10.921 &-0.074 & -0.116 & -0.199 \\ 
Smoking in the past 30 days & TOB30 & 0.026 & 0.025 & 0.008 & 0.049 & 0.278 & 0.026 & 0.000 & 0.000 & 0.000 \\ \hline
\end{tabular}
\end{sidewaystable}

To evaluate the balance of covariates between the female and male groups, we use the standardized mean difference (SD1), and define the standardized variance difference (SD2) and standardized third central moment difference (SD3) as follows:
\[
\text{SD2} = \frac{2(s_1^2-s_0^2)}{s_1^2+s_0^2}, \quad \text{SD3} = \frac{2^{\frac{3}{2}}(\nu_1-\nu_0)}{(s_1^2+s_0^2)^{\frac{3}{2}}},
\]
where $s_m^2$ and $\nu_m$ are the sample variance and the third sample central moment for male participants $(m = 1)$ and female participants $(m = 0)$. 
Table \ref{tab:realapp2}(a) lists the means, SD1, SD2, and SD3 values before and after matching. 
The results indicate that the covariates are not balanced before matching. For instance, the SD1, SD2, and SD3 values for HATT are 0.422, 1.534, and 4.741, respectively. 
After matching, the covariates still appear to be poorly balanced. Specifically, the SD3 values for the covariates HATT, ANG1, ANG2, and ALC1 remain larger than 4.

We apply pHT, CD, CM, CNN1, CNN5, CMST1, CMST5, and $D5$ to the paired data after matching. Among these, only our proposed test ($D5$) produces a small $p$-value of 0.020. By contrast, the $p$-values for pHT, CD, CM, CNN1, CNN5, CMST1, and CMST5 are all larger than 0.05 (Table \ref{tab:intro-eg}), indicating that these tests cannot reject the null hypothesis of balanced covariates at the 0.05 significance level.
To investigate which result is more reliable, we apply the paired $t$-test to each of the 14 covariates. 
The smallest $p$-value from the paired $t$-tests is 0.001 (Table \ref{tab:app-cov1}). 
Using the Bonferroni correction ($0.05/14=0.004>0.001$), the paired $t$-test rejects the null hypothesis at the 0.05 significance level, supporting the result of $D5$.

\begin{table}[!h]
\centering
\caption{The $p$-values of pHT, CD, CM, CNN1, CNN5, CMST1, CMST5, and $D5$}\label{tab:intro-eg}
\begin{tabular}{@{}cccccccc@{}}\hline
pHT & CD & CM & CNN1 & CNN5 & CMST1 & CMST5 & $D5$\\
0.062 & 0.076 & 0.888 & 0.461 & 0.247 & 0.337 & 0.239 & 0.020 \\ \hline
\end{tabular}
%\vspace{0.2cm}
\end{table}

\begin{table}[!h]
\centering
\caption{The $p$-values from the paired $t$-test applied to each of the 14 covariates separately}\label{tab:app-cov1}
 \begin{tabular}{@{}ccccccc@{}}\hline
  AGE & BPSYS &  TOB30 & SMOKYRS & ANG1 & ALC1 & DEP2YRS \\
   0.791 & 0.740 & 1.000 & 0.718 & 0.158 & 0.566 & 1.000 \\
  BMI & BPDIAS & TOB100 & HATT & ANG2 & ALC2 & PACKSPER \\ 
  0.290 & 0.725 & 0.670  & \textbf{0.001}  &  \textbf{0.011} & 1.000 & 0.679 \\\hline
\end{tabular}
\end{table}

\subsection{Comparison of neuropsychologic performances for pair-matched data}\label{sec:realapp3}
Another question of interest is to compare the neuropsychological performances of well-matched female and male participants. To facilitate the matching process, we focus on data from the initial visit and restrict the analysis to white participants who meet the following criteria: no history of stroke, transient ischemic attack, atrial fibrillation, heart problems, angioplasty/endarterectomy/stent, active depression in the last two years, diabetes, brain trauma, Parkinson’s disease, seizures, alcohol abuse, or other neurological or psychiatric disorders. Additionally, participants must have no family history of frontotemporal lobar degeneration mutation, Alzheimer’s disease mutation, or cognitive impairment in first-degree relatives.
For this analysis, we consider only four covariates: AGE, BMI, SMOKYRS, and TOB30 (Table \ref{tab:realapp2}(b)). Using the same matching method as described in Section \ref{sec:realapp2}, we match 76 male participants with 76 female participants from a pool of 120 female participants.

The means, SD1, SD2, and SD3 before and after matching are presented in Table \ref{tab:realapp2}(b). The results show that the covariates after matching are much more balanced compared to those in Section \ref{sec:realapp2}. We apply the proposed test $D5$ alongside other existing tests to the matched covariates. Since the $p$-values for pHT, CD, CM, CNN1, CNN5, CMST1, CMST5, and $D5$ are all greater than 0.05 (Table \ref{tab:real-pairmatch2}), none of these tests reject the null hypothesis at the 0.05 significance level.
Additionally, we examine the $p$-values of the paired $t$-test for each variable. As shown in Table \ref{tab:real-pairmatch2-bct}, the null hypothesis cannot be rejected after applying the Bonferroni correction. Based on these results, we conclude that the covariates are jointly reasonably balanced.
\begin{table}[!h]
\centering
\caption{The $p$-values of pHT, CD, CM, CNN1, CNN5, CMST1, CMST5, and $D5$}\label{tab:real-pairmatch2}
\begin{tabular}{@{}cccccccc@{}}\hline
pHT & CD & CM & CNN1 & CNN5 & CMST1 & CMST5 & $D5$ \\
0.510 & 0.500 & 0.906 & 0.281 & 0.818 & 0.082 & 0.225 & 0.116 \\ \hline
\end{tabular}
%\vspace{0.2cm}
\end{table}
\begin{table}[!h]
\centering
\caption{The $p$-values from the paired $t$-test for variables AGE, BMI, SMOKYRS and TOB30}\label{tab:real-pairmatch2-bct}
\begin{tabular}{@{}cccc@{}}\hline
AGE & BMI & SMOKYRS & TOB30 \\
0.982 & 0.604 & 0.638 & 1.000  \\ \hline
\end{tabular}
%\vspace{0.2cm}
\end{table}

After obtaining well-matched participants, we apply the paired Hotelling's $T^2$ test (pHT) and the proposed test ($D5$) to assess whether neuropsychological performances differ between the female and male groups. 
We focus on the Neuropsychological Battery Summary Scores from Form C1 of the UDS, which includes 22 neuropsychological measurement variables. A description of these 22 variables is provided in \ref{app:variable}.
The $p$-values for pHT and $D5$ are 0.053 and 0.009, respectively. While pHT does not reject the null hypothesis of equal neuropsychological performances between the female and male groups at the 0.05 significance level, the proposed test $D5$ does. 
To further investigate, we apply the paired $t$-test to each of the 22 covariates. The results in \ref{app3:t-test} show that the smallest $p$-value from the paired $t$-test is 0.001. 
Using the Bonferroni correction ($0.05/22=0.002>0.001$), the paired $t$-test also rejects the null hypothesis at the 0.05 significance level.
These findings suggest that the proposed test $D5$ is more reliable { and may reflect an underlying association between sex and Alzheimer's disease.}

\subsection{Comparison of neuropsychologic performances between two visits}\label{sec:realapp1}
  % the data set we obtained consists of 6315 participants that had no less than 4 visits (initial + at least 3 follow-up). 
%Given the richness of the dataset, a lot of problems can be studied. 

We study the participants' neuropsychological performance over time, considering 22 neuropsychology measurement variables as described in Section \ref{sec:realapp3} (\ref{app:variable}). Participants are grouped into three categories based on the CDR\textcircled{R} Dementia Staging Instrument at their first visits: no dementia (Group I), very mild dementia (Group II), and mild dementia (Group III).
We test whether their neuropsychological performances over five years differ from those at their initial visits. After removing missing data, the sample sizes ($n$) for the three groups are 1746, 543, and 41, respectively.

\begin{table}[!h]
\centering
\caption{Test results of pHT and $D5$ (\( p \)-values less than 0.05 are highlighted in bold)}\label{tab:realsample}
\begin{tabular}{crcrcrc}\hline
& \multicolumn{2}{c}{Group I} & \multicolumn{2}{c}{Group II} & \multicolumn{2}{c}{Group III} \\
& Statistic & $p$-value  & Statistic & $p$-value  & Statistic & $p$-value  \\ 
 pHT & 26.14 & \textbf{$<$1e-4} & 11.92 & \textbf{$<$1e-4}  & 1.43 & 0.2194 \\
$D5$ & 38.90 & \textbf{$<$1e-4} & 34.70 & \textbf{$<$1e-4} & 12.30 & \textbf{0.0021} \\ \hline
\end{tabular}
\end{table}

Table \ref{tab:realsample} presents the results of the paired Hotelling's $T^2$  test (pHT) and the proposed test $D5$.
We first examine the results for Group I.  
Both pHT and $D5$ strongly reject the null hypothesis, with extremely small $p$-values. 
These results indicate a significant difference in neuropsychological measures between the two visits. Similarly, for Group II, both tests also reject the null hypothesis, suggesting a significant change in neuropsychological performance over time.

For Group III, the relatively small sample size reduces the power of the tests. The paired Hotelling's $T^2$ test does not reject the null hypothesis, whereas the proposed test $D5$ rejects the null hypothesis with a small $p$-value. To determine which result is more reliable, we examine the data in greater detail. Specifically, we perform paired $t$-tests on each of the 22 neuropsychological measurement variables. The results show that 13 out of the 22 variables have $p$-values less than 0.05, with the smallest $p$-value being 0.001 (Table \ref{tab:t-test} in \ref{app1:t-test}). Using the Bonferroni correction ($0.05/22=0.002>0.001$), the paired $t$-test rejects the null hypothesis at the 0.05 significance level, supporting the result from $D5$.
%For Group III, test (ii), the sample size is even smaller. Since the $p$-value of $D$ is a little larger than 0.05, $D$ rejects the null hypothesis at 0.1 significance level, but not at 0.05 significance level. Hence, there is some evidence that the neuropsychologic measures in six years differ from those in their initial visits for Group III, while more observations are needed to make a confident conclusion.

These results reveal that, for participants across all three groups, neuropsychological measures after several years become significantly different from those at their initial visits. This finding suggests that researchers should closely monitor changes in cognitive performance, even for participants with no dementia or (very) mild dementia.

\section{Discussion}\label{sec:discuss}
\subsection{Paired mean test and paired variance test}
When a mean difference exists, within-sample observations tend to cluster more closely, causing the numbers of within-sample edges, \(R_1\) and \(R_2\), to exceed their expected values. To detect mean differences, we use the standardized value of \(R_1 + R_2\).
In contrast, for variance differences, observations from the sample with smaller variance cluster closely, resulting in a higher-than-expected number of within-sample edges. Meanwhile, observations from the sample with larger variance tend to connect across samples due to the curse of dimensionality,  the difference \(R_1 - R_2\) becomes more pronounced. We use the absolute value of the standardized \(R_1 - R_2\) to detect variance differences.
These phenomena also align with those observed for unpaired data (see \cite{chen2017new} for details).
%This insight motivates the development of two specific test statistics tailored to detect either mean differences or variance differences.

{ 
Based on these insights, we can construct two test statistics tailored to detect mean differences and variance differences, respectively. To test for a difference in means, we use the statistic
\begin{equation}\label{eq:Zm}
D_m = \frac{R_1 + R_2 - \bE(R_1 + R_2)}{\sqrt{\var(R_1 + R_2)}},
\end{equation}
and reject the null hypothesis at significance level \(\alpha\) if \(D_m > C_m(\alpha)\). To test for a difference in variances, we use the statistic
\begin{equation}\label{eq:Zs}
D_s = \frac{R_1 - R_2 - \bE(R_1 - R_2)}{\sqrt{\var(R_1 - R_2)}},
\end{equation}
and reject the null hypothesis if \(|D_s| > C_s(\alpha)\). 
As established in Theorem~\ref{th:asym} and stated formally in Corollary~\ref{corollary2}, both \(D_m\) and \(D_s\) converge in distribution to the standard normal under the null hypothesis. For ease of implementation, we set the critical values as \(C_m(\alpha) = \Phi^{-1}(1 - \alpha)\) and \(C_s(\alpha) = \Phi^{-1}(1 - \alpha/2)\), where \(\Phi^{-1}(b)\) denotes the $b$-quantile of the standard normal distribution.}

\begin{corollary}\label{corollary2}
Under Conditions \ref{cond1}, \ref{cond2}, and \ref{cond3}, as $N\rightarrow\infty$, 
\[
\begin{psmallmatrix}
D_m\\D_s
\end{psmallmatrix}\longrightarrow\cN\left(\begin{psmallmatrix}
0 \\ 0
\end{psmallmatrix}, \begin{psmallmatrix}
1 & 0 \\
0 & 1
\end{psmallmatrix}\right)
\]
in distribution under the paired-comparison permutation null distribution.
\end{corollary}

With the analytic expressions for expectations and variances in Theorem \ref{th:expression0}, we can further investigate the relationship among the statistics $D$, $D_m$, and $D_s$. 
\begin{proposition}\label{th:Zg}
The following relationships hold:
\[
D = D_m^2 + D_s^2 \text{ and } \cov(D_m,D_s)=0.
\]
\end{proposition}
\begin{proof}
Let
  \begin{equation*}
    \mathbf{R}=\begin{pmatrix}
      R_{1}-\bE(R_{1}) \\
      R_{2}-\bE(R_{2})
    \end{pmatrix}, \quad
    \check{\mathbf D}=\begin{pmatrix}
                   D_{m} \\
                   D_{s}
                 \end{pmatrix}=\begin{pmatrix}
                                 \frac{1}{\sqrt{\var(R_1+R_2)}} & \frac{1}{\sqrt{\var(R_1+R_2)}} \\
                                 \frac{1}{\sqrt{\var(R_1-R_2)}} & -\frac{1}{\sqrt{\var(R_1-R_2)}}
                               \end{pmatrix}\mathbf{R}\triangleq \mathbf B\mathbf{R}.
  \end{equation*}
  It is straightforward to verify that $\mathbf B$ is invertible. From the definition of $D$ (Equation (\ref{eq:Zg})), $D$ can be expressed as
  \[
 D=\mathbf{R}^\top\mathbf{\Sigma}^{-1}_R\mathbf{R}=(\mathbf B^{-1}\check{\mathbf D})^\top\mathbf\Sigma^{-1}_R(\mathbf B^{-1}\check{\mathbf D})=\check{\mathbf D}^\top(\mathbf B\mathbf\Sigma_R \mathbf B^\top)^{-1}\check{\mathbf D}.
  \]
By substituting $\mathbf B$ and $\mathbf{\Sigma}_R$, we have $\mathbf B\mathbf\Sigma_R \mathbf B^\top=\begin{psmallmatrix}
1 & 0 \\
0 & 1
\end{psmallmatrix}$.  
Thus, $D=\check{\mathbf D}^\top\check{\mathbf D}=D_m^2+D_s^2$, and since the covariance term is zero, $\cov(D_m,D_s)=0$.
\end{proof}

\subsection{{Sensitivity analysis for unobserved confounding in matched studies}}\label{sec:discuss-sensitivity}

{In practice, unobserved covariates may exist when matching subjects. To assess the influence of potential unobserved confounding, we adopt the sensitivity analysis framework in \cite{rosenbaum2010design}.
Let \(\Gamma\) denote the sensitivity parameter, representing the ratio of treatment-assignment odds between two matched subjects.  Let \(\pi_i\) be the probability of assigning treatment to one subject in pair \(i\). 
When no unobserved confounding is present ($\Gamma=1$), the treatment assignment probability for each subject equals ($\pi_i=0.5$).
When unobserved confounding exists (\(\Gamma > 1\)), the treatment assignment probability satisfies
\[
\frac{1}{\Gamma + 1} \;\;\leq\;\; \pi_i \;\;\leq\;\; \frac{\Gamma}{\Gamma + 1}.
\]

To conduct the sensitivity analysis, we introduce bias into the treatment assignment probabilities when generating permutations in the paired-comparison permutation framework. Specifically, for each pair \(i\), the observed treatment is retained with probability $\pi_i$, and the treatment assignment is reversed with probability $1 - \pi_i$.
In this scenario, the analytic expressions for the expectations and variances presented in Theorem~\ref{th:expression0} no longer hold. { While an analytic characterization would be valuable for sensitivity analysis, the complex form of the test statistic poses significant challenges, which we leave for future research.} Nonetheless, we can still estimate an adjusted $p$-value using weight-adjusted permutations.
To achieve this, we repeat the weight-adjusted permutation process for a large number of iterations ($B$) to compute the adjusted permutation expectations ($\bE_{\text{adj}}(R_1)$ and $\bE_{\text{adj}}(R_2)$) and variance $\mathbf{\Sigma}_{\text{adj}}$. Using these adjustments, we calculate the adjusted test statistic:  
\begin{equation*}
D_{\text{adj}} = \begin{pmatrix}
R_1 - \bE_{\text{adj}}(R_1) \\
R_2 - \bE_{\text{adj}}(R_2)
\end{pmatrix}^\top
\mathbf{\Sigma}_{\text{adj}}^{-1}
\begin{pmatrix}
R_1 - \bE_{\text{adj}}(R_1) \\
R_2 - \bE_{\text{adj}}(R_2)
\end{pmatrix}.
\end{equation*}
Similarly, the permuted test statistics are calculated as:
\begin{equation*}
D_{\text{perm}, b} = \begin{pmatrix}
R_{1,b} - \bE_{\text{adj}}(R_1) \\
R_{2,b} - \bE_{\text{adj}}(R_2)
\end{pmatrix}^\top
\mathbf{\Sigma}_{\text{adj}}^{-1}
\begin{pmatrix}
R_{1,b} - \bE_{\text{adj}}(R_1) \\
R_{2,b} - \bE_{\text{adj}}(R_2)
\end{pmatrix},
\end{equation*}
for \(b = 1, \dots, B\).
The adjusted permutation \(p\)-value for a given sensitivity parameter \(\Gamma\) is computed as:
\[
p_{\text{perm}, \Gamma} 
= \frac{1}{B} \sum_{b=1}^B 1\bigl(D_{\text{perm}, b} > D_{\text{adj}}\bigr).
\]
{ Specifically, to approximate the worst-case $p$-value, we randomly generate \(\pi_i \in \{1/(\Gamma+1), \Gamma / (\Gamma + 1)\}\) for each \(\Gamma\) and repeat the procedure over \(M\) runs. The maximum adjusted permutation \(p\)-value across these runs is then taken as the worst-case estimate.}

In Section \ref{sec:realapp3}, unobserved confounding may be present when matching male and female participants. To evaluate the impact, we perform a sensitivity analysis using various sensitivity parameters \(\Gamma = 1, 1.5, \dots, 3.5\). For each scenario, we conduct \(M = 500\) runs and \(B = 1000\) weight-adjusted permutations.
Table~\ref{tab:real-sensitivity} presents the maximum permutation \(p\)-values. The results indicate that the proposed test \(D5\) remains robust, consistently achieving statistically significant results (\(p < 0.05\)) for \(\Gamma \leq 2\). This demonstrates the resilience of $D5$ to potential unobserved confounding.

\begin{table}[!h]
\centering
\caption{ {Maximum adjusted permutation \(p\)-values for pHT and \(D5\) across 500 runs for varying sensitivity parameters \(\Gamma\)}}\label{tab:real-sensitivity}
%\begin{tabular}{@{}ccccccc@{}}\hline
%\(\Gamma\) & 1 & 1.5 & 2 & 2.5 & 3 & 3.5 \\ 
%pHT & 0.053 & 0.065 & 0.068 & 0.078 & 0.085 & 0.092 \\
%\(D5\) & 0.009 & 0.027 & 0.033 & 0.039 & 0.048 & 0.068 \\ \hline
%\end{tabular}
\begin{tabular}{@{}ccccccc@{}}\hline
\(\Gamma\) & 1 & 1.5 & 2 & 2.5 & 3 & 3.5 \\ 
pHT & 0.053 & 0.075 & 0.097 & 0.123 & 0.165 & 0.200 \\
\(D5\) & 0.009 & 0.029 & 0.048 & 0.053 & 0.084 & 0.099 \\ \hline
\end{tabular}
\end{table}
}

\subsection{{ Fixed-ratio and variable-ratio matching examination}}
The proposed test can be extended to fixed-ratio matching. Suppose each subject in the treatment group is matched with a fixed number, \(k\), of subjects in the control group. For each matching \(i \in \{1, \dots, n\}\), let \(X_i\) denote the covariates for subject \(i\) in the treatment group, and let \(Y_{i1}, \dots, Y_{ik}\) denote the covariates for the corresponding \(k\) subjects in the control group.
Similar to the null hypothesis~\eqref{eq:prob} considered in this paper, we assume that any permutation of \(\bigl(X_i, Y_{i1}, \dots, Y_{ik}\bigr)\) has the same distribution under the null hypothesis. In this context, the permutation null distribution is defined by assigning a probability of \(\bigl(k+1\bigr)^{-n}\) to each of the \(\bigl(k+1\bigr)^{n}\) possible ways of “assigning one observation from \(\{X_i, Y_{i1}, \dots, Y_{ik}\}\) to sample~1 and the remaining observations to sample~2.”

{ 
Extending the proposed framework to variable-ratio matching introduces additional challenges. In this setting, each matched set \(i \in \{1, \ldots, n\}\) may contain a different number of observations, denoted by \(\{X_i, Y_{i1}, \ldots, Y_{ik_i}\}\). As a result, it is generally inappropriate to assume that the joint distribution of \((X_i, Y_{i1}, \ldots, Y_{ik_i})\) is identical across all matched sets.
One approach to address this complexity is to adopt an exchangeability-based perspective. Under this view, the null hypothesis can be formulated such that, within each matched set, the joint distribution remains invariant under permutations of the observations. Under this formulation, the permutation null distribution assigns a uniform probability of \(\prod_{i=1}^n (k_i + 1)^{-1}\) to each of the \(\prod_{i=1}^n (k_i + 1)\) possible ways of assigning one observation from \(\{X_i, Y_{i1}, \ldots, Y_{ik_i}\}\) to sample~1 and the remaining \(k_i\) observations to sample~2.

Within the two extended permutation frameworks discussed above, it is then possible to derive analytic expressions for expectations and variances analogous to those in Theorem~\ref{th:expression0}, and to conduct asymptotic analysis using arguments similar to those in Section~\ref{sec:asym}. A detailed investigation of this extension is left for future work.
}

\subsection{{Degeneracy issue}}

{%From a theoretical perspective, Conditions~\ref{cond2} and \ref{cond3} are relatively mild, as sufficient conditions are provided in Remark~\ref{rem:conds} to ensure that they are satisfied. 
In practice, to determine whether \((R_1, R_2)\) is degenerate, the condition number of \(\Sigma_R\) can be computed. If its value exceeds a specified threshold (e.g., \(10^6\)), we recommend using only the standardized \(R_1\) as the test statistic. The asymptotic distribution of this statistic is provided in the following corollary.

\begin{corollary}
Under Condition~\ref{cond1}, and assuming either Condition~\ref{cond2} or Condition~\ref{cond3} holds, as \(N \to \infty\),
\[
D_R \triangleq \frac{\bigl(R_1 - \mathbb{E}(R_1)\bigr)^2}{\mathrm{Var}(R_1)} 
\;\longrightarrow\; \chi^2_1
\]
in distribution under the paired-comparison permutation null distribution, where \(\chi^2_1\) denotes the chi-squared distribution with one degree of freedom.
\end{corollary}

To test the null hypothesis at the significance level $\alpha$, we  reject it if \(D_R > \chi^2_1(1-\alpha)\), where \(\chi^2_1(1-\alpha)\) denotes the \((1-\alpha)\)-quantile of the \(\chi^2_1\) distribution.}

\section{Conclusion}\label{sec:conclude}
%In the regime of two-sample comparison, tests based on a graph constructed on observations by utilizing similarity information among them are gaining attention due to their flexibility and good performances for high-dimensional/non-Euclidean data. However, all existing graph-based tests are inapplicable for paired data due to the dependent paired observations. 

Paired data commonly arise in numerous scenarios, including independent paired data from pair matching and non-independent paired data from paired designs. In many modern datasets, the number of measurements is often comparable to, or even larger than, the number of pairs. 
We propose a new non-parametric test for paired data using a graph-based two-sample testing framework.
Because existing graph-based tests do not fully consider the paired structure and may underperform in some settings, we introduce the paired-comparison permutation null distribution and develop the statistic $D$. 
Our numerical experiments—both in simulations and in real applications—demonstrate that $D$ performs effectively for assessing covariate balance in pair matching and for testing non-independent paired data. It exhibits high power across a wide variety of settings, including shape, location, and/or scale alternatives.
Under the paired-comparison permutation null distribution, we also derive the asymptotic distribution of our new statistic $D$. 

As an advantage inherited from graph-based tests, the proposed method can also be applied to multi-type complex data, such as non-Euclidean data, provided a suitable distance metric can be defined on the observations.

\section*{Acknowledgement}
{Jingru Zhang was supported in part by National Natural Science Foundation of China (NSFC 12401388) and Shanghai Pujiang Program (Grant No. 23PJ1401100).}
Hao Chen was supported in part by NSF awards DMS-1513653, DMS-
1848579, and DMS-2311399.
Xiao-Hua Zhou was supported in part by the National Science Foundation of China (NSFC 81773546 and NSFC 12026606). 
{The NACC database is funded by NIA/NIH Grant U24 AG072122. NACC data are contributed by the NIA-funded ADRCs: P30 AG062429 (PI James Brewer, MD, PhD), P30 AG066468 (PI Oscar Lopez, MD), P30 AG062421 (PI Teresa Gomez-Isla, MD), P30 AG066509 (PI Thomas Grabowski, MD), P30 AG066514 (PI Mary Sano, PhD), P30 AG066530 (PI Helena Chui, MD, Arthur Toga, PhD), P30 AG066507 (PI Marilyn Albert, PhD), P30 AG066444 (PI David Holtzman, MD), P30 AG066518 (PIs Lisa Silbert, MD, Kevin Duff, PhD), P30 AG066512 (PI Thomas Wisniewski, MD), P30 AG066462 (PI Scott Small, MD), P30 AG072979 (PI David Wolk, MD), P30 AG072972 (PIs Charles DeCarli, MD, Rachel Whitmer, PhD), P30 AG072976 (PI Andrew Saykin, PsyD), P30 AG072975 (PI Julie Schneider, MD, MS), P30 AG072978 (PI Ann McKee, MD), P30 AG072977 (PI Robert Vassar, PhD), P30 AG066519 (PI Joshua Grill, PhD), P30 AG062677 (PIs Brad Boeve, MD, Ronald Petersen, MD, PhD), P30 AG079280 (PI Jessica Langbaum, PhD), P30 AG062422 (PI Gil Rabinovici, MD), P30 AG066511 (PI Allan Levey, MD, PhD), P30 AG072946 (PI Linda Van Eldik, PhD), P30 AG062715 (PI Sanjay Asthana, MD, FRCP), P30 AG072973 (PI Russell Swerdlow, MD), P30 AG066506 (PIs Glenn Smith, PhD, ABPP, David Lowenstein, PhD, Ranjan Duara, MD), P30 AG066508 (PIs Stephen Strittmatter, MD, PhD, Christopher Van Dyck, MD), P30 AG066515 (PI Victor Henderson, MD, MS), P30 AG072947 (PI Suzanne Craft, PhD), P30 AG072931 (PI Henry Paulson, MD, PhD), P30 AG066546 (PIs Sudha Seshadri, MD, Gladys Maestre, MD, PhD), P30 AG086401 (PI Erik Roberson, MD, PhD), P30 AG086404 (PI Gary Rosenberg, MD), P30 AG086403 (PI Angela Jefferson, PhD), P30 AG072958 (PIs Heather Whitson, MD, Gwenn Garden, MD, PhD), P30 AG072959 (PI Jagan Pillai, MD, PhD), P30 AG092752 (Ihab Hajjar, MD, MS).}

\appendix

\gdef\thesection{Appendix \Alph{section}}

\section{{ An example illustrating the distinction between the null {hypotheses} \eqref{eq:prob} and \eqref{eq:testold}}}\label{toyexample}
Let \((X_i, Y_i),~i=1, \dots, n,\) be paired data, where \(X_i\) takes values in a discrete integer state space:
\[
\mathcal{S} = \{A, A+1, \dots, B\}, \qquad L := B - A + 1.
\]
Define
\[
Y_i = f(X_i) :=
\begin{cases}
X_i + \alpha(X_i), & X_i \le B - \alpha(X_i), \\
X_i - (L - \alpha(X_i)), & X_i > B - \alpha(X_i),
\end{cases}
\]
where \(\alpha: \mathcal{S} \to \{0,1, 2, \dots, L-1\}\) is a shift function chosen so that \(f\) is a bijection (i.e., one-to-one and onto), but not equal to its own inverse.

Assume each \(X_i\) is drawn independently and uniformly from \(\mathcal{S}\), i.e., \(P(X_i = k) = 1/L\) for all \(k \in \mathcal{S}\). Since \(f\) is a bijection, the distribution of \(Y_i = f(X_i)\) is also uniform: \(P(Y_i = k) = P(X_i = f^{-1}(k)) = 1/L\). Thus, the marginal distributions are equal, i.e., \(F_X = F_Y\). However, because \(f \ne f^{-1}\), the joint distribution of \((X_i, Y_i)\) is not symmetric under coordinate swapping, and hence \(F_1 \ne F_2\).

As a concrete example,
suppose the joint distribution of \((X_i, Y_i)\) is:
\[
\begin{aligned}
& P(X_i = 1, Y_i = 3) = P(X_i = 2, Y_i = 4) = P(X_i = 3, Y_i = 5) \\
& = P(X_i = 4, Y_i = 1) = P(X_i = 5, Y_i = 2) = 1/5.
\end{aligned}
\]
Here, \(P(X_i = k) = P(Y_i = k) = 1/5\) for \(k = 1, 2, \dots, 5\), so \(F_X = F_Y\) holds. However, for instance, \(P(X_i = 1, Y_i = 3) = 1/5 \ne P(X_i = 3, Y_i = 1) = 0\), which shows that the joint distribution is not symmetric. Therefore, \(F_1 \ne F_2\).

\section{Proof of Theorem \ref{th:asym}}\label{pf:asym}
\begin{proof}
The proof of Theorem \ref{th:asym} relies on Stein's method. Consider the sum of the form $W=\sum_{i\in\cJ}\xi_i,$ where $\cJ$ is an index set and $\xi_i$ are random variables with $\bE(\xi_i)=0,$ and $\bE(W^2)=1.$ The following assumption restricts the dependence between $\{\xi_i:i\in\cJ\}$.

  \begin{assumption}\label{a}
    [\citet{chen2005stein}, p. 17] For each $i\in\cJ$ there exists $S_i\subset T_i\subset \cJ$ such that $\xi_i$ is independent of $\xi_{S_i^c}$ and $\xi_{S_i}$ is independent of $\xi_{T_i^c}$.
  \end{assumption}

  We will use the following theorem.

  \begin{theorem}\label{th:app}
    [\citet{chen2005stein}, Theorem 3.4] Under Assumption \ref{a}, we have
    \[
    \sup_{h\in Lip(1)}|\bE h(W)-\bE h(Z_0)|\leq\delta
    \]
    where $Lip(1)=\{h:\mathbb{R}\rightarrow\mathbb{R},~\parallel h'\parallel\leq 1\}$, $Z_0$ has $\cN(0,1)$ distribution and
    \[
    \delta=2\sum_{i\in\cJ}(\bE|\xi_i\eta_i\theta_i|+|\bE(\xi_i\eta_i)|\bE|\theta_i|)+\sum_{i\in\cJ}\bE|\xi_i\eta_i^2|
    \]
    with $\eta_i=\sum_{j\in S_i}\xi_j$ and $\theta_i=\sum_{j\in T_i}\xi_j,$ where $S_i$ and $T_i$ are defined in Assumption \ref{a}.
  \end{theorem}

Let
$
\bE(R_1) = \bE(R_2) \triangleq \mu$, $\var(R_1) = \var(R_2) \triangleq \sigma^2$,  $\cov(R_1,R_2) \triangleq \sigma_{12}$, 
\[
W_1 = \frac{R_1-\mu}{\sigma} \quad\text{ and }\quad W_2 = \frac{R_2-\mu}{\sigma}.
\]
Under the conditions of Theorem \ref{th:asym}, as $N\rightarrow\infty$, we can prove the following results:
  
(1) $(W_1,W_2)$ follows a bivariate Gaussian distribution,

(2) $|\lim_{N\rightarrow\infty}\corr(W_1,W_2)|<1.$

 To prove (1), by Cram$\acute{\text{e}}$r-Wold device, we only need to show that $W=a_1W_1+a_2W_2$ is asymptotically Gaussian distributed for any combination of $a_1,a_2$ such that $\var(W)>0$.

  We first define more notations. For any edge $e=(u,v)$ of $G_1$, $i.e.~uv\in\cJ=\{uv:u<v,(u,v)\in G_1\},$ let
  \[
  R_{uv}^{(1)}=I(g_u=g_v=1),\quad d_{uv}^{(1)}=\bE(R_{uv}^{(1)})=\frac{1}{4},\quad\xi_{uv}^{(1)}=\frac{R_{uv}^{(1)}-d_{uv}^{(1)}}{\sigma},
  \]
  \[
  R_{uv}^{(2)}=I(g_u=g_v=0),\quad d_{uv}^{(2)}=\bE(R_{uv}^{(2)})=\frac{1}{4},\quad\xi_{uv}^{(2)}=\frac{R_{uv}^{(2)}-d_{uv}^{(2)}}{\sigma}.
  \]
  Thus, 
  \[
  W_1 = \sum_{uv\in\cJ}\xi_{uv}^{(1)},\quad W_2 = \sum_{uv\in\cJ}\xi_{uv}^{(2)},
  \]
  \[
  W = a_1W_1+a_2W_2 = \sum_{uv\in\cJ}(a_1\xi_{uv}^{(1)}+a_2\xi_{uv}^{(2)})\triangleq\sum_{uv\in\cJ}\xi_{uv},
  \]
  where $\xi_{uv} = a_1\xi_{uv}^{(1)}+a_2\xi_{uv}^{(2)}.$
  
  We introduce following index sets to satisfy Assumption \ref{a}. For $uv\in\cJ$, let
\begin{align*}
& S_{uv} \doteq A_e = \{(i,j)\in G_1:i\in \{u,u^*,v,v^*\} \text{ or } j\in\{u,u^*,v,v^*\}\}, \\
& T_{uv} \doteq B_e = \cup_{\tilde e\in A_e}A_{\tilde e}.
\end{align*}
   Let $a=\max\{|a_1|,|a_2|\}$. Since $R_{uv}^{(1)}\in\{0,1\}$ and $d_{uv}^{(1)}=1/4$, we have $|\xi_{uv}^{(1)}|\leq 3/(4\sigma)$. Similarly, $|\xi_{uv}^{(2)}|\leq 3/(4\sigma)$ and we have  $|\xi_{uv}|\leq 3a/(2\sigma)$.
   
   Hence, 
   \[
   \sum_{j\in S_{uv}}|\xi_j|\leq \frac{3a}{2\sigma}|A_e| \quad \text{ and }\quad\sum_{j\in T_{uv}}|\xi_j|\leq \frac{3a}{2\sigma}|B_e|,
   \]
   where $e = (u,v)$.
   For $i=uv\in\cJ$, the terms $\bE|\xi_i\eta_i\theta_i|, |\bE(\xi_i\eta_i)|\bE|\theta_i|$ and $\bE|\xi_i\eta_i^2|$ are all bounded by
   \[
   \frac{27a^3}{8\sigma^3}|A_e||B_e|.
   \]
   
   So we have $ \sup_{h\in Lip(1)}|\bE h(\widetilde{W})-\bE h(Z_0)|\leq\delta$ with $\widetilde{W}=W/\sqrt{\var(W)}$, $Z_0\sim {\cal N}(0,1)$, and
  \begin{align*}
    \delta = & \frac{1}{\sqrt{\var^3(W)}}\Bigg\{2\sum_{i\in\cJ}(\bE|\xi_i\eta_i\theta_i|+|\bE(\xi_i\eta_i)|\bE|\theta_i|)+\sum_{i\in\cJ}\bE|\xi_i\eta_i^2|\Bigg\} \\
    \leq & \frac{27a^3}{\sigma^3\sqrt{\var^3(W)}}\sum_{e\in G_1}|A_e||B_e|.
    \end{align*}
    Since $27a^3/\sqrt{\var^3(W)}$ is of constant order and $\sigma = O(N^{0.5\gamma})$, when $\sum_{e\in G_1}|A_e||B_e|=o(N^{1.5\gamma})$, we have $\delta\rightarrow 0$ as $N\rightarrow \infty$.
    
    Then we prove (2). We have
\begin{align*}
& \lim_{N\rightarrow\infty}\corr(W_1,W_2) = \lim_{N\rightarrow\infty}\frac{\sigma_{12}}{\sigma^2} = \lim_{N\rightarrow\infty}\frac{\sigma_{12}}{\sqrt{\sigma_{12}^2+4b_1b_2}},
\end{align*}
where 
\[
b_1 = \frac{1}{16}\sum_{(i,j)\in G_1}\left\{1+I(i^*\in G_{1,j^*})-I( i \in G_{1,j^*} ) - I( j \in G_{1,i^*} )\right\}=\frac{1}{16}(|G_1|+2C_1-2C_2)
\]
and 
\[
b_2 = \frac{1}{16}\sum_{i=1}^n(|G_{1,i}|-|G_{1,i^*}|)^2.
\]
Since $b_1=O(N^\gamma)$ and $b_2=O(N^\gamma)$, we obtain $|\lim_{N\rightarrow\infty}\corr(W_1,W_2)|<1$.
\end{proof}

\section{Proof of Theorem \ref{th:consistency}}\label{pf:consistency}
\begin{proof}
Let the density functions of the two multivariate distributions be $f$ and $g$. Following the approach in \citet{henze1999multivariate}, we have
\[
\frac{R_1}{N}\rightarrow \frac{k}{2}\int \frac{f^2(x)}{f(x)+g(x)} dx, \quad\frac{R_2}{N}\rightarrow \frac{k}{2}\int \frac{g^2(x)}{f(x)+g(x)} dx, 
\]
almost surely.
Let 
\[
\delta_1 = \lim_{N\rightarrow\infty}\frac{R_1-\mu_1}{N},\quad \delta_2 = \lim_{N\rightarrow\infty}\frac{R_2-\mu_2}{N},
\]
\[
r_1 = \lim_{N\rightarrow\infty} \frac{|G_1|+2C_1-2C_2}{N},\quad r_2 = \lim_{N\rightarrow\infty} \frac{\sum_{i=1}^n(|G_{1,i}|-|G_{1,i^*}|)^2}{N}.
\]
Then
\begin{align*}
\lim_{N\rightarrow\infty}\frac{D}{N} = & \lim_{N\rightarrow\infty}\left(\frac{R_1-\mu_1}{N},\frac{R_2-\mu_2}{N}\right)\left(\frac{\Sigma_R}{N}\right)^{-1}\left(\frac{R_1-\mu_1}{N},\frac{R_2-\mu_2}{N}\right)^\top \\
= & \frac{4}{r_1r_2}(\delta_1,\delta_2)\begin{pmatrix}
r_1+r_2 & r_2-r_1 \\
r_2-r_1 & r_1+r_2
\end{pmatrix}
(\delta_1,\delta_2)^\top \\
= & \frac{4}{r_1r_2}\left[r_1(\delta_1-\delta_2)^2+r_2(\delta_1+\delta_2)^2\right].
\end{align*}
We next show that $\delta_1+\delta_2>0$ when $f$ and $g$ differ on a set of
positive measure.
Note that 
\begin{align*}
\frac{2(\delta_1+\delta_2)}{k} = & \lim_{N\rightarrow\infty}\int\frac{f^2(x)+g^2(x)}{f(x)+g(x)}dx-\frac{|G_1|}{Nk} \\
= & \int\frac{f^2(x)+g^2(x)}{f(x)+g(x)}dx-1+\lim_{N\rightarrow\infty}\frac{|G|-|G_1|}{Nk} \\
\geq & \int\frac{f^2(x)+g^2(x)}{f(x)+g(x)}dx-1.
\end{align*}
Since 
\begin{align*}
\int\frac{f^2(x)+g^2(x)}{f(x)+g(x)}dx-1 & = \int\frac{f^2(x)+g^2(x)}{f(x)+g(x)}-f(x)dx = \int\frac{g(x)(g(x)-f(x))}{f(x)+g(x)}dx \\
& = \int\frac{f^2(x)+g^2(x)}{f(x)+g(x)}-g(x)dx = \int\frac{f(x)(f(x)-g(x))}{f(x)+g(x)}dx, 
\end{align*}
we have 
\begin{align*}
\int\frac{f^2(x)+g^2(x)}{f(x)+g(x)}dx-1 = \frac{1}{2}\int\frac{(f(x)-g(x))^2}{f(x)+g(x)}dx.
\end{align*}
So $\delta_1+\delta_2$ is strictly positive when $f$ and $g$ differ on a set of positive
measure.
\end{proof}

\section{{Additional simulation results}}
\label{app:addsimu}
{As part of the sensitivity analysis, we also apply the proposed test using the 1-MST, 5-MST, and 10-MST constructed under the \(L_1\) distance, denoted by \(\widetilde{D1}\), \(\widetilde{D5}\), and \(\widetilde{D10}\), respectively.  
Table \ref{tab:addsimu1} presents the results for \(\widetilde{D1}\), \(\widetilde{D5}\), and \(\widetilde{D10}\) under the settings described in Section \ref{sec:simu-pairmatch}, and
 Table \ref{tab:addsimu2} presents the results for settings in Section \ref{sec:simu-correlated}.}

\begin{table}[!h]
\centering
\caption{{Proportion of trials (out of 1000) in which the test rejects covariate balance at the 0.05 significance level}}\label{tab:addsimu1}
\begin{tabular}{@{}lcccccccc@{}}\hline
& A1 & A2 & A3 & A4 & B1 & B2 & B3 & B4 \\
 $\widetilde {D1}$ & 0.049 & 0.046 & 0.621 & 0.634 & 0.054 & 0.045 & 0.461 & 0.457 \\
 $\widetilde {D5}$ & 0.042 & 0.055 & 0.921 & 0.924 & 0.050 & 0.047 & 0.813 & 0.820 \\
 $\widetilde {D10}$ & 0.039 & 0.047 & 0.945 & 0.945 & 0.057 & 0.054 & 0.873 & 0.877 \\ \hline
\end{tabular}
%\vspace{0.2cm}
\end{table}

\begin{table}[!h]
\centering
\caption{{Estimated power at the 0.05 significance level based on 1000 runs.}}\label{tab:addsimu2}
(a) Data from the same family of distributions

\smallskip

\begin{tabular}{llcccccc}\hline\smallskip
& & \multicolumn{3}{c}{(i) $\nu_1\neq\nu_2$, $\Gamma_1=\Gamma_2$} & \multicolumn{3}{c}{(ii) $\nu_1\neq\nu_2$, $\Gamma_1\neq\Gamma_2$}\\
& $d$ & 50 & 100 & 1000 & 50 & 100 & 1000\\
multivariate normal & $\widetilde {D1}$ & 0.140 & 0.127 & 0.136 & 0.840 & 0.834 & 0.996 \\ 
& $\widetilde {D5}$ & 0.602 & 0.662 & 0.670 & 0.963 & 0.956 & 1.000 \\
& $\widetilde {D10}$ & 0.819 & 0.878 & 0.876 & 0.994 &0.994 & 1.000 \\
multivariate $t$ & $\widetilde {D1}$ & 0.204 & 0.225 & 0.238 & 0.779 & 0.798 & 0.940 \\ 
& $\widetilde {D5}$ & 0.703 & 0.707 & 0.430 & 0.931 &0.902 & 0.894 \\
& $\widetilde {D10}$ & 0.799 & 0.811 & 0.434 & 0.965 &0.935 & 0.888 \\
multivariate log-normal & $\widetilde {D1}$ & 0.924 &0.957 & 0.996 & 0.998 & 0.997 & 1.000 \\
& $\widetilde {D5}$ & 0.932 & 0.963 & 1.000 & 0.998 &0.997 & 1.000 \\
& $\widetilde {D10}$ & 0.974 & 0.989 & 1.000 & 1.000 &1.000 & 1.000 \\ \hline
\end{tabular}

\vspace{0.3cm}
(b) Data from different families of distributions

\smallskip

\begin{tabular}{llcccccc}\hline\smallskip
& & \multicolumn{3}{c}{$\Delta=\bzero_d$} & \multicolumn{3}{c}{$\Delta=0.2\bone_d$}\\
& $d$ & 50 & 100 & 1000 & 50 & 100 & 1000\\ 
$\tau_i\stackrel{iid}{\sim}\text{normal}$ & $\widetilde {D1}$ &  0.675 & 0.858 & 1.000& 0.712 & 0.894 & 1.000\\
& $\widetilde {D5}$ & 0.783 & 0.918 & 1.000& 0.869 & 0.976 & 1.000\\
& $\widetilde {D10}$ & 0.829 & 0.957 & 1.000&0.930 & 0.994 & 1.000\\
$\tau_i\stackrel{iid}{\sim}\text{skew normal}$ & $\widetilde {D1}$ & 0.580 & 0.784 & 1.000&0.566 & 0.805 & 1.000\\ 
& $\widetilde {D5}$ & 0.692 & 0.870 & 1.000& 0.733 & 0.932 & 1.000\\
& $\widetilde {D10}$ & 0.728 & 0.913 & 1.000&0.803 & 0.957 & 1.000\\
$\tau_i\stackrel{iid}{\sim}\text{Laplace}$ & $\widetilde {D1}$ & 0.354 & 0.530 & 1.000& 0.386 & 0.636 & 1.000\\ 
& $\widetilde {D5}$ & 0.426 & 0.646 & 1.000& 0.601 & 0.884 & 1.000\\
& $\widetilde {D10}$ & 0.466 & 0.698 & 1.000&0.727 & 0.964 & 1.000\\\hline
\end{tabular}

\end{table}

\section{More details of the real application} \label{app:application}
\subsection{Neuropsychology measure variables}\label{app:variable}
Table \ref{tab:variable} lists the neuropsychology measure variables in the real application of Alzheimer's disease.  
%Table \ref{tab:t-test} lists the $p$-values of applying the paired $t$-test to each of the variables separately for Group III, test (i).

\begin{table}[!h]
\centering
\caption{A description of neuropsychology measure variables}\label{tab:variable}
\begin{tabular}{ll}\hline\smallskip
Variable & Description  \\
MMSEORDA & Orientation subscale score --- Time \\ 
MMSEORLO & Orientation subscale score --- Place \\ 
PENTAGON & Intersecting pentagon subscale score \\
NACCMMSE & Total MMSE score \\
LOGIMEM & Total number of story units recalled from this current test administration \\
MEMUNITS & Logical Memory IIA --- Delayed --- Total number of story units recalled \\
MEMTIME & Logical Memory IIA --- Delayed --- Time elapsed since Logical Memory \\
& IA --- Immediate \\
DIGIF & Digit span forward trials correct \\
DIGIFLEN & Digit span forward length \\
DIGIB & Digit span backward trials correct \\
DIGIBLEN & Digit span backward length \\
ANIMALS & Total number of animals named in 60 seconds \\
VEG & Total number of vegetables named in 60 seconds \\
TRAILA & Trail Making Test Part A --- Total number of seconds to complete \\
TRAILARR & Part A --- Number of commission errors \\
TRAILALI & Part A --- Number of correct lines \\
TRAILB & Trail Making Test Part B --- Total number of seconds to complete \\
TRAILBRR & Part B --- Number of commission errors \\
TRAILBLI & Part B --- Number of correct lines \\
WAIS & WAIS-R Digit Symbol \\
BOSTON & Boston Naming Test (30) --- Total score \\
COGSTAT & Per clinician, based on the neuropsychological examination, the subject's \\
&  cognitive status is deemed \\ \hline
\end{tabular}
\end{table}

\subsection{Additional results of Section \ref{sec:realapp3}}\label{app3:t-test}

\begin{table}[!h]
\centering
\caption{The $p$-values of applying the paired $t$-test to each of the 22 neuropsychology measurement variables separately for pair-matched male and female participants}\label{tab:app-cov2}
 \begin{tabular}{lc}\hline\smallskip
Variable & $p$-value \\
MMSEORDA & 0.184 \\
MMSEORLO & 0.849 \\
PENTAGON & 0.321 \\
NACCMMSE & 0.346 \\
LOGIMEM & \textbf{0.025} \\
MEMUNITS & \textbf{0.009} \\
MEMTIME & 0.893 \\
DIGIF & 0.904 \\
DIGIFLEN & 0.691 \\
DIGIB & 0.088 \\
DIGIBLEN & 0.358 \\
ANIMALS & 0.640 \\
VEG & \textbf{0.001} \\
TRAILA & \textbf{0.016} \\
TRAILARR & 0.810 \\
TRAILALI & 1.000 \\
TRAILB & 0.607 \\
TRAILBRR & 0.743 \\
TRAILBLI & 0.321 \\
WAIS & \textbf{0.040} \\
BOSTON & 0.257 \\
COGSTAT & 0.556 \\ \hline
\end{tabular}
\end{table}

Table \ref{tab:app-cov2} shows the $p$-values of the paired $t$-test when it is applied to each of the 22 neuropsychology measurement variables separately.
Since the smallest $p$-value of the paired t-test is 0.001, the paired $t$-test would reject the null hypothesis at 0.05 significance level $(0.05/22 = 0.002>0.001)$ by the Bonferroni correction.

\subsection{Additional results of Section \ref{sec:realapp1}}\label{app1:t-test}

\begin{table}[!h]
\centering
\caption{The $p$-values of applying the paired $t$-test to each of the variables separately for the mild dementia group, the initial visit versus the visit after five years}\label{tab:t-test}
 \begin{tabular}{lc}\hline\smallskip
Variable & $p$-value \\
MMSEORDA & 0.068 \\
MMSEORLO & 0.173 \\
PENTAGON & 0.534 \\
NACCMMSE & \textbf{0.017} \\
LOGIMEM & 0.171 \\
MEMUNITS & 0.399 \\
MEMTIME & 0.469 \\
DIGIF & \textbf{0.031} \\
DIGIFLEN & \textbf{0.014} \\
DIGIB & \textbf{0.027} \\
DIGIBLEN & \textbf{0.031} \\
ANIMALS & \textbf{0.035} \\
VEG & \textbf{0.010} \\
TRAILA & \textbf{0.036} \\
TRAILARR & \textbf{0.049} \\
TRAILALI & 0.418 \\
TRAILB & \textbf{0.004} \\
TRAILBRR & 0.227 \\
TRAILBLI & \textbf{0.013} \\
WAIS & \textbf{0.001} \\
BOSTON & \textbf{0.005} \\
COGSTAT & 0.643 \\ \hline
\end{tabular}
\end{table}
Table \ref{tab:t-test} shows the $p$-values of the paired $t$-test when it is applied to each of the 22 neuropsychology measurement variables separately.
Since the smallest $p$-value of the paired t-test is 0.001, the paired $t$-test would reject the null hypothesis at 0.05 significance level $(0.05/22 = 0.002>0.001)$ by the Bonferroni correction.

\clearpage
\bibliographystyle{elsarticle-harv} 
\bibliography{paired_test}


\end{document}